\documentclass{llncs}

\usepackage{mathtools}
\usepackage{amsmath}
\usepackage{amssymb}
\usepackage{xcolor}
\usepackage{adjustbox}
\usepackage{array, makecell}
\usepackage{framed}
\usepackage{setspace}
\usepackage{graphicx}
\usepackage{bbm}
\usepackage{amsfonts}
\usepackage[breaklinks=true]{hyperref}
\usepackage{breakcites}

\usepackage[linesnumbered,lined]{algorithm2e}

\SetCommentSty{mycommfont}
\DontPrintSemicolon
\usepackage{newfloat}
\DeclareFloatingEnvironment[fileext=lop]{Algorithm}

\newcommand{\myfrac}[3][0pt]{\genfrac{}{}{}{0}{\raisebox{-#1}{$#2$}}{\raisebox{-#1}{$#3$}}}
\newcommand{\fdfrac}[2]{\mbox{\footnotesize$\displaystyle\frac{#1}{#2}$}}

\DeclareMathOperator*{\argmin}{arg\,min}

\newcommand{\X}{\mathcal{X}}
\newcommand{\veps}{\varepsilon}
\newtheorem{fact}{Fact}
\newcommand{\C}{\mathcal{C}}
\newcommand{\LK}{{\tt List-k-service}}

\definecolor{mycolor}{rgb}{1, 0.0, 0.0}
\newcommand{\dishant}[1]{\textcolor{mycolor}{}}

\allowdisplaybreaks

\begin{document}

\title{FPT Approximation for Constrained Metric $k$-Median/Means}
%\lv{
%\subtitle{(Full Version)}
%}
%
%\titlerunning{}
% abbreviated title (for running head)
\author{Dishant Goyal \and Ragesh Jaiswal \and Amit Kumar}
%
%\authorrunning{Jaiswal et al.}
% abbreviated author list (for running head)
%\tocauthor{Ragesh Jaiswal}
%
\institute{Department of Computer Science and Engineering, 
Indian Institute of Technology Delhi.\thanks{\email{\{Dishant.Goyal, rjaiswal, amitk\}@cse.iitd.ac.in}}
}
{\def\addcontentsline#1#2#3{}\maketitle}
%\maketitle     % typeset the title of the contribution
%%\thispagestyle{empty}

\begin{abstract}
The Metric $k$-median problem over a metric space $(\X, d)$ is defined as follows: given a set  $L \subseteq \X$ of {\em facility locations} and a set $C \subseteq \X$ of {\em clients}, open a set $F \subseteq L$ of $k$ facilities such that the total service cost, defined as $\Phi(F, C) \coloneqq \sum_{x \in C} \min_{f \in F} d(x, f)$, is minimised. 
The metric $k$-means problem is defined similarly using squared distances (i.e., $d^2(., .)$ instead of $d(., .)$). 
In many applications there are additional constraints that any solution needs to satisfy.
For example, to balance the load among the facilities in resource allocation problems, a capacity $u$ is imposed on every facility. 
That is, no more than $u$ clients can be assigned to any facility. 
This problem is known as the \emph{capacitated} $k$-means/$k$-median problem. 
Likewise, various other applications have different constraints, which give rise to different {\em constrained} versions of the problem such as \emph{$r$-gather}, \emph{fault-tolerant}, \emph{outlier} $k$-means/$k$-median problem.
Surprisingly, for many of these constrained problems, no constant-approximation algorithm is known. 
Moreover, the \emph{unconstrained} problem itself is known~\cite{capacitated:FPT_2018_Byrka} to be $\mathsf{W[2]}$-hard when parameterized by $k$.
We give FPT algorithms with constant approximation guarantee for a range of constrained $k$-median/means problems. 
For some of the constrained problems, ours is the first constant factor approximation algorithm whereas for others, we improve or match the approximation guarantee of previous works.
We work within the unified framework of Ding and Xu~\cite{Ding_and_Xu_15} that allows us to simultaneously obtain algorithms for a range of constrained problems.
In particular, we obtain a $(3+\veps)$-approximation and $(9+\veps)$-approximation for the constrained versions of the $k$-median and $k$-means problem respectively in FPT time.
In many practical settings of the $k$-median/means problem, one is allowed to open a facility at any client location, i.e., $C \subseteq L$.
For this special case, our algorithm gives a $(2+\veps)$-approximation and $(4+\veps)$-approximation for the constrained versions of $k$-median and $k$-means problem respectively in FPT time.
Since our algorithm is based on simple sampling technique, it can also be converted to a constant-pass log-space streaming algorithm. 
In particular, here are some of the main highlights of this work:
\begin{enumerate}
\item For the uniform capacitated $k$-median/means problems our results matches previously known results of Addad {\em et al.}~\cite{capacitated:FPT_2019_vincent}.
\item For the $r$-gather $k$-median/means problem (clustering with lower bound on the size of clusters), our FPT approximation bounds are better than what was previously known. 
\item Our approximation bounds for the fault-tolerant, outlier, and uncertain versions is better than all previously known results, albeit in FPT time.
\item For certain constrained settings such as chromatic, $l$-diversity, and semi-supervised $k$-median/means, we obtain the first constant factor approximation algorithms to the best of our knowledge.
\item Since our algorithms are based on a simple sampling based approach, we also obtain constant-pass log-space streaming algorithms for most of the above-mentioned problems.
\end{enumerate}
\end{abstract}

\section{Introduction}\label{section:introduction}

The metric $k$-means and $k$-median problems are similar. We combine the discussion of these problems by giving a definition of the $k$-service problem that encapsulates both these problems. 

\begin{definition}[$k$-service problem]
Let $(\X, d)$ be a metric space, $k > 0$ be any integer and $\ell \geq 0$ be any real number.
Given a set $L \subseteq \X$ of feasible facility locations, and a set $C \subseteq \X$ of clients, find a set $F \subseteq L$ of $k$ facilities that minimises the total service cost: $\Phi(F,C) \equiv \sum_{j \in C} \min_{i \in F}d^{\ell}(i,j)$.
\end{definition}
Note that the $k$-service problem is also studied with respect to a more general cost function $\sum_{j \in C} \min_{i \in F} \delta(i,j)$, where $\delta(i,j)$ denotes the cost of assigning a client $j \in C$ to a facility $i \in F$.
We  consider the special case $\delta(i,j) \equiv d^{\ell}(i,j)$. 
For $\ell = 1$, the problem is known as the $k$-\emph{median} problem and for $\ell = 2$, the problem is known as the $k$-\emph{means} problem.
The above definition is motivated by the \emph{facility location} problem and differs from it in two ways.  
First, in the facility location problem, one is allowed to open any number of facilities. 
Second, one has to pay for an additional facility establishment cost for every open facility. 
Thus the $k$-service problem is basically the facility location problem for a fixed number of facilities.

The $k$-service problem can also be viewed as a clustering problem, where the goal is to group the objects that are similar to each other. Clustering algorithms are commonly used in data mining, pattern recognition, and  information retrieval~\cite{akjain_survey_1}. 
However, the notion of a cluster differs for different applications. 
For example, some applications consider a cluster as a dense region of points in the data-space~\cite{dbscan_1996,OPTICS_1999}, while others consider it as a highly connected sub-graph of a graph~\cite{graph_clustering_2000}. 
Likewise, various models have been developed in the past that capture the clustering properties in different ways~\cite{survey_on_all_algorithms_2015}. 
The $k$-means and $k$-median problems are examples of the \emph{center-based} clustering model.
In this model, the objects are mapped to the points in a metric space such that the distance between the points captures the degree of dissimilarity between them. 
In other words, the closer the two points are, the more similar they are to each other. 
In order to measure the quality of a clustering, a center (known as the cluster representative) is assigned to each cluster and the cost is measured based on the distances of the points to their respective cluster centers. 
Then the problem objective is to obtain a clustering with the minimum cost.
To view the $k$-median instance as a clustering instance, consider the client set as a set of data points and the facility locations as the feasible centers. 
In a feasible solution, the clients which are assigned to the same facility are considered a part of the same cluster and the corresponding facility act as their cluster center. 
During our discussion, we will use the term \emph{center} and \emph{facility} interchangeably.
Similarly, we can view the $k$-means problem as a clustering problem where the cost is measured with respect to the squared distances.

%The \emph{Lloyd's} method~\cite{lloyds}, popularly known as the $k$-means algorithm, has been the most popular clustering method due to its speed and simplicity. However, the method performs exponentially bad on some instances of the problem. Since the real-world scenarios are not so bad, it is widely used for practical purposes. A similar and better algorithm, i.e., the $k$-means++ algorithm~\cite{kmeanspp} provides a $\Theta(\log k)$-approximation guarantee in expectation. Moreover, this is the best algorithm known so far, as long as practicality is concerned. 

%Theoretically, the problem has been well studied.
%The problem has various variations due to which hardness results differ for different variants. 
Various variants of the $k$-median/means problem have been studied in the clustering literature.
For example, the Euclidean $k$-means problem (where $C \subseteq L = \mathbb{R}^{d}$) is  $\mathsf{NP}$-hard even for a fixed $k$ or a fixed dimension $d$~\cite{das08,aloise09,MNV09,V09}. 
This opens the question of designing a PTAS (polynomial-time approximation schemes) for the problem when either the number of clusters or the dimension is fixed. 
Indeed, various PTASs are known under such conditions~\cite{kumar02,Feldman07,chen09,jks14,friggstad16,addad16}. 
%A more challenging task is to design a PTAS when $k$ and $d$ are not fixed. 
In general, it is known that the problem can not be approximated within a constant factor, unless $\mathsf{P} = \mathsf{NP}$~\cite{acks15,addad19}.

The hardness results in the previous paragraph was for Euclidean setting.
These problems may be harder in general metric spaces which is indeed what has been shown. 
The metric $k$-median problem is hard to approximate within a factor of $(1+2/e)$, and the metric $k$-means problem is hard to approximate within a factor of $(1+8/e)$~\cite{GUHA1998,kamal02}. 
On the positive side, various constant-factor approximation algorithms are known for the $k$-means (and $k$-median) problems in the metric and Euclidean settings ~\cite{kanungo02,charikar99,naveen02,anupam_gupta_08,svensson13,byrka15,Svensson17}. 
%All these mentioned results hold for the unconstrained $k$-means/$k$-median problem, and 
Improving these bounds is not the goal of this paper. 
Instead, we undertake the task of improving/obtaining approximation bounds of a more general class of problems called the \emph{constrained $k$-means/$k$-median} problem. 
Let us see what these problems are and why they are important.
%theoretically and practically. 

For many real-world applications, the classical (unconstrained) $k$-means and $k$-median problems do not entirely capture the desired clustering properties. 
For example, consider the popular $k$-\emph{anonymity} principle~\cite{rgather:k_anonymity_2002_Sweeney}. 
The principle provides anonymity to a public database while keeping it meaningful at the same time. 
One way to achieve this is to cluster the data release only partial information related to the clusters obtained.
Further, to protect the data from the \emph{re-identification} attacks, the clustering should be done in such a way that each cluster gets at least $r$ data-points. This method is popularly known as \emph{$r$-gather} clustering~\cite{rgather:k_center_2010_Aggarwal} (see the formal definition in Table~\ref{table:1}). Likewise, various other applications impose a specific set of constraints on the clusters. 
Such applications have been studied extensively. A survey on these applications is mentioned in Section 1.1 of ~\cite{Ding_and_Xu_15}). We collectively mention these problems in Table~\ref{table:1} and their known approximation results in Table~\ref{table:2}. 
We discuss these problems and their known results in detail in Section~\ref{sec:constrained_problems} of the Appendix.

\begin{table}[h]
\begin{adjustbox}{width=\columnwidth,center}
\centering
\begin{tabular}{|l|l|l|}
\hline
\# & {\bf Problem} & {\bf Description} \\ \hline
1. & \makecell[l]{$r$-gather $k$-service problem* \\ $(r, k)$-{\tt GService}} & \makecell[l]{Find clustering $\mathcal{C} = \{C_1, ..., C_k\}$ with minimum $\Psi^{*}(\mathcal{C})$ \\such that for all $i$, $|C_i| \geq r_{i}$} \\ \hline
2. & \makecell[l]{$r$-Capacity $k$-service problem* \\ $(r, k)$-{\tt CaService}} & \makecell[l]{Find clustering $\mathcal{C} = \{C_1, ..., C_k\}$ with minimum $\Psi^{*}(\mathcal{C})$ \\such that for all $i$, $|C_i| \leq r_{i}$} \\ \hline
3. & \makecell[l]{$l$-Diversity $k$-service problem \\ $(l, k)$-{\tt DService}} & \makecell[l]{Given that every client has an associated colour, \\find a clustering $\mathcal{C} = \{C_1, ..., C_k\}$ with minimum $\Psi^{*}(\mathcal{C})$ \\such that for all $i$, the fraction of points sharing the \\same colour inside $C_i$ is $\leq \frac{1}{l}$} \\ \hline
4. & \makecell[l]{Chromatic $k$-service problem \\ $k$-{\tt ChService}} & \makecell[l]{Given that every client has an associated colour,  \\find a clustering $\mathcal{C} = \{C_1, ..., C_k\}$ with minimum $\Psi^{*}(\mathcal{C})$ \\such that for all $i$, $C_i$ should not have any two \\points with the same colour.} \\ \hline
5. & \makecell[l]{Fault tolerant $k$-service problem \\ $(l, k)$-{\tt FService}} & \makecell[l]{Given a value $l_{p}$ for every client, find a clustering \\ $\mathcal{C} = \{C_1, ..., C_k\}$ and a set $F$ of $k$ centers, such that\\ the sum of service cost of the points to $l_{p}$ of nearest\\ centers out of $F = \{f_{1},f_{2}, \dotsc,f_{k}\}$, is minimised.} \\ \hline
6. & \makecell[l]{Semi-supervised $k$-service problem \\ $k$-{\tt SService}} & \makecell[l]{Given a target clustering $\mathcal{C}' = \{C_1', ..., C_k'\}$ and constant $\alpha$ \\find a clustering
$\mathcal{C} = \{C_1, ..., C_k\}$ and a center set $F$, such \\that the cost $\overline{\Psi}(F,\C) \coloneqq \alpha \cdot \Psi(F,\C) + (1-\alpha) \cdot Dist(\C', \C)$ \\ is minimised. $Dist$ denotes the set-difference distance.} \\ \hline
7. & \makecell[l]{Uncertain $k$-service problem \\ $k$-{\tt UService}} & \makecell[l]{Given a discrete probability distribution for every client, \\ i.e., for a point $p \in C$ there is a set $D_{p} = \{ p_{1},\dotsc,p_{h}\}$ \\ such that $p$ takes the value $p_{i}$ with probability $t_{p}^{i}$ \\ and $\sum_{i =1}^{h} t_{p}^{i} \leq 1$. Find a clustering
$\mathcal{C} = \{C_1, ..., C_k\}$ \\ so that the expected cost of $\Psi^{*}(\C)$ is minimized.} \\ \hline
8. & \makecell[l]{Outlier $k$-service problem* \\ $(k,m)$-{\tt OService}} & \makecell[l]{Find a set $Z\subseteq C$ of size $m$ \\ and a clustering $\mathcal{C'} = \{C'_1, ..., C'_{k}\}$ of the set $C' \coloneqq C \setminus Z$,\\ such that $\Psi^{*}(\C')$ is minimized.}  \\ \hline
\end{tabular}
\end{adjustbox}
\vspace*{0.5mm}
\caption{Constrained $k$-service problems with efficient partition algorithm (see Section~4 and~5.3 in \cite{Ding_and_Xu_15} and references therein). The (*)marked problems were not discussed in~\cite{Ding_and_Xu_15}. We mention their partition algorithms in Section~\ref{section: partition_algm}. 
}
\label{table:1}
\end{table}

An important distinction between the constrained problems and their unconstrained counterparts is the idea of \emph{locality}. 
In simple words, the \emph{locality} property says that the points which are close to each other should be part of the same cluster. 
This property holds for the unconstrained version of the problem.
However, this may not necessarily hold for many of the constrained versions of the problem where minimising clustering cost is not the only requirement.
To understand this, consider a center-set $F = \{f_{1},f_{2},\dotsc,f_{k}\}$ and let $\{C_1, ..., C_k\}$ denote the clustering of the dataset such that the cost function gets minimised. 
That is, $C_i$ contain all the points for which $f_i$ is the closest center in the set $F$.
Note that the clustering $\{C_1, ..., C_k\}$ just minimises the distance based cost function and may not satisfy any additional constraint that the clustering may need to satisfy in a constrained setting.
In a constrained setting we may need an algorithm that, given a center-set $\{f_1, ..., f_k\}$ as input, outputs a clustering $\{\bar{C}_1, ..., \bar{C}_k\}$ which in addition to minimising $\sum_i \sum_{x \in \bar{C}_i}d^{\ell}(x, f_i)$ also satisfies certain clustering constraints.
Such an algorithm is called a {\em partition algorithm}.
In the unconstrained setting, the partition algorithm simply assigns points to closest center in $F$.
%We can obtain the optimal clustering corresponding to $F$ simply by assigning the clients to the closest center. Such a method is known as the \emph{partition algorithm} of a problem.Using it, we can obtain an optimal clustering if an optimal center-set is known to us. For the (unconstrained) $k$-means/$k$-median problem, the partition algorithm is a simply a \emph{Voronoi partitioning} algorithm. 
However, designing such an efficient partition algorithm for the constrained versions of the problem is a non-trivial task. 
Ding and Xu~\cite{Ding_and_Xu_15} gave partition algorithms for all the problems mentioned in Table~\ref{table:1} (see Section 4 and 5.3 of~\cite{Ding_and_Xu_15}). 
Though these algorithms were specifically designed for the Euclidean space, they can be generalized to any metric space.
We will see that such a partition algorithm is crucial in the design of our FPT algorithms.

The partition algorithm gives us a way for going from center-set to clustering. What about the reverse direction?
Given a clustering $\C = \{ C_{1}, C_{2}, \dotsc, C_{k}\}$, can we find a center set that gives minimum clustering cost?
The solution to this problem is simple. 
Construct a complete weighted bipartite graph $G = (V_{l},V_{r},E)$, where a vertex in $V_{l}$ corresponds to a facility location in $L$, and a vertex in $V_{r}$ corresponds to a cluster $C_{j} \in \C$. 
The weight on an edge $(i,j) \in V_{l} \times V_{r}$ is equal to the cost of assigning the cluster $C_{j}$ to the $i^{th}$ facility, i.e., $\sum_{x \in C_{j}} d^{\ell}(x,i)$. 
Then we can easily obtain an optimal assignment by finding the minimum cost perfect matching in the graph $G$. 
Let us denote the minimum cost by $MCPM(\mathcal{C}, L)$
Thus, it is sufficient to output an optimal clustering for a constrained $k$-service instance. 
In fact, all problems in Table~\ref{table:1} only requires us to output an optimal clustering for the problem.

%As you can see from Table~\ref{table:2}, for many constrained problems, we do not have any constant-approximation algorithm (corresponding to the $k$-means and $k$-median objectives). The difficulty in designing these algorithms arise due to the extra set of constraints imposed on a clustering. 

Ding and Xu~\cite{Ding_and_Xu_15} suggested the following unified framework for considering any constrained $k$-means/$k$-median problem by modeling an arbitrary set of constraints using feasible clusterings.
Note that they studied the problem in the Euclidean space where $C \subseteq L = \mathbb{R}^{d}$ whereas we study the problem in general metric space where $L$ and $C$ are discrete and separate sets. 
We will use a few more definitions to define the problem. 
A $k$-center-set is a set of $k$ distinct elements from $L$ and for any $k$-center-set $F = \{f_1, ..., f_k\}$ and a clustering $\mathcal{C} = \{C_1, ..., C_k\}$, we will use the cost function:
$$\Psi(F, \mathcal{C}) \equiv \min_{\textrm{permutation } \pi} \left\{ \sum_{i=1}^{k} \sum_{x \in C_i} d^{\ell}(x, f_{\pi(i)}) \right\}.$$
%F\or that, we give a similar definition for the problem as follows.

\begin{definition}[Constrained $k$-service problem]\label{definition:constrained_k_service}
Let $(\X, d)$ be a metric space, $k > 0$ be any integer and $\ell \geq 0$ be any real number.
Given a set $L \subseteq \X$ of feasible facility locations, a set $C \subseteq \X$ of clients, and a set $\mathbb{C}$ of feasible clusterings, find a clustering $\C = \{C_{1}, C_{2}, \dotsc, C_{k}\}$ in $\mathbb{C}$, that minimizes the following objective function:
$\Psi^{*}(\C) ~ \equiv ~ \min\limits_{\textrm{$k$-center-set $F$}} \Psi(F, \mathcal{C})$.
\end{definition}
Note that $\Psi(F, \mathcal{C})$ is $MCPM(\mathcal{C}, L)$, the minimum cost perfect matching as discussed earlier.
The key component of the above definition is the set of feasible clusterings $\mathbb{C}$. 
Using this, we can define any constrained version of the problem.
Note that $\mathbb{C}$ can have an exponential size. However, for many problems it can be defined concisely  using a simple set of mathematical constraints. 
For example, $\mathbb{C}$ for the $r$-gather problem can be defined as $\mathbb{C} \coloneqq \{ \C ~ \mid ~ \textrm{for every cluster } C_{i} \in \C$, $|C_{i}| \geq r_{i} \}$, where $\C = \{C_{1},C_{2},\dotsc,C_{k} \}$ is a partitioning of the client set. 
Note that we consider the \emph{hard assignment} model for the problem.
That is, one cannot open more than one facility at a location. 
It differs from the \emph{soft assignment} model where one can open multiple facilities at a location. 
The soft version can be stated in terms of the hard version -- by allowing $L$ to be a multi-set and creating $k$-copies for each location in $L$. 
%This distinction is not relevant for the unconstrained $k$-service problem because if we pick a center twice, we can always throw a copy away to make the assignment hard. 
%However, in the $r$-capacity $k$-service problem, if we open multiple facilities at the same location, we cannot remove any facility because the leftover clients cannot be added to other clusters due to their limited capacity. 
It has been observed that the soft-assignment models are easier and allow better approximation guarantees than the hard-assignment models~\cite{capacitated:kcenter_2012_khuller,capacitated:kmedian_2017_Li}. 
For our discussion, we will call a center-set a \emph{soft center-set} if it contains facility location multiple times, otherwise we call it a \emph{hard center-set}. 
In fact, a soft center-set is a multi-set. We will avoid using the term multi-set to keep our discussion simple.

As observed in past works~\cite{Ding_and_Xu_15,bjk18}, any constrained version of $k$-median/means can be solved using a partition algorithm for this version and a solution to a very general ``list'' version of the clustering problem which we discuss next.
Let us define this problem which we call the \emph{list $k$-service} problem
\footnote{This notion of list version of the clustering problem was implicitly present in the work of Ding and Xu~\cite{Ding_and_Xu_15}. Bhattacharya et al.~\cite{bjk18} formalized this as the {\em list $k$-means problem}.}.
This will help us solve the constrained $k$-service problem. 

\begin{definition}[List $k$-service problem]\label{definition:list_k_service}
Let $\alpha$ be a fixed constant.
Let $\mathcal{I} = (L,C,k,d,\ell)$ be any instance of the $k$-service problem and let $\C = \{C_{1}, C_{2}, \dotsc,C_{k}\}$ be an arbitrary clustering of the client set $C$.
The goal of the problem is: given $\mathcal{I}$, find a list $\mathcal{L}$ of $k$-center-sets (i.e., each element of the list is a set of $k$ distinct elements from $L$) such that, with probability at least $(1/2)$, there is a $k$-center-set $F$ such that $\Psi(F,\C) \leq \alpha \cdot \Psi^{*}(\C)$.
%where $\Psi(F,\C) \coloneqq \min\limits_{\textrm{permutation } \pi} \sum_{i = 1}^{k} \sum_{x \in C_{i}} d^{\ell}(x,f_{\pi(i)})$.
\end{definition}

Note that the clustering algorithm in the above setup does not get access to the clustering $\mathcal{C}$ and yet is supposed to find good centers (constant $\alpha$ approximation) for this clustering. 
Given this, it is easy to see that finding a single set of $k$ centers that are good for $\mathcal{C}$ is not possible. 
However, finding a reasonably small list of $k$-center-sets such that at least one of the $k$-center-sets in the list is good may be feasible.
This is main realization behind the formulation of the list version of the problem.
The other reason is that since the target clustering is allowed to be a completely arbitrary partition of the client set $C$, we can use the solution of the list $k$-service problem to solve any constrained $k$-service problem as long as there is a partition algorithm.
The following theorem combines the list $k$-service algorithm and the partition algorithm for a constrained version of the problem to produce a constant-approximation algorithm this problem.

\begin{theorem}\label{theorem: list_alpha_approx}
Let $\mathcal{I} = (C,L,k,d,\ell,\mathbb{C})$ be any instance of any constrained $k$-service problem and let $A_{\mathbb{C}}$ be the corresponding partition algorithm.
Let $B$ be an algorithm for the list $k$-service problem that runs in time $T_B$ for instance $(C,L,k,d,\ell)$.
%Suppose we are given a list $\mathcal{L}$, then we can design 
There is an algorithm that, with probability at least $1/2$, outputs a clustering $\C \in \mathbb{C}$, which is an $\alpha$-approximation for the constrained $k$-service instance. 
The running time of the algorithm is $O(T_B + |\mathcal{L}| \cdot T_A)$, where $T_A$ is the running time of the partition algorithm.
\end{theorem}

\begin{proof}
The algorithm is as follows. We first run algorithm $B$ to obtain a list $\mathcal{L}$.
For every $k$-center-set in the list, the algorithm runs the partition algorithm $A_{\mathbb{C}}$ on it. 
Then the algorithm outputs that $k$-center-set that has the minimum clustering cost. 
Let $F'$ be this $k$-center-set and $\C'$ be the corresponding clustering. 
We claim that $(F',\C')$ is an $\alpha$-approximation for the constrained $k$-service problem.

Let $\C^{*}$ be an optimal solution for the constrained $k$-service instance $(C,L,k,d,\ell,\mathbb{C})$ and $F^{*}$ denote the corresponding $k$-center-set. 
By the definition of the list $k$-service problem, with probability at least $1/2$, there is a $k$-center-set $F$ in the list $\mathcal{L}$, such that $\Psi(F,\C^{*}) \leq \alpha \cdot \Psi(F^{*},\C^{*})$. 
Let $\C = A_{\mathbb{C}}(F) \in \mathbb{C}$ be the clustering corresponding to $F$. 
Thus, $\Psi(F,\C) \leq \Psi(F,\C^{*}) \leq \alpha \cdot \Psi(F^{*},\C^{*})$. Since $F'$ gives the minimum cost clustering in the list, we have $\Psi(F',C') \leq \Psi(F,C)$. 
Therefore, $\Psi(F',C') \leq \alpha \cdot \Psi(F^{*},\C^{*})$. 

Since, the algorithm runs a partition procedure for every center set in the list, the running time of this step is $|\mathcal{L}| \cdot T_A$. 
Picking a minimum cost clustering from the list takes $O(|\mathcal{L}|)$ time. Hence the overall running time is $O(T_B + |\mathcal{L}| \cdot T_A)$.
\end{proof}

Now suppose we are given a list $\mathcal{L}$ of size $g(k)$ (for some function $g$) and a partition algorithm for the problem with the polynomial running time.
Then by Theorem~\ref{theorem: list_alpha_approx}, we get an FPT algorithm for the constrained $k$-service problem. 
Since for many of the constrained $k$-service problems there exists efficient partition algorithms, it makes sense to design an algorithm for the list $k$-service problem that outputs a list of size at most $g(k)$. 
We will design such an algorithm in Section~\ref{section:list_k_service} of this paper.
We also need to make sure that the partition algorithms for constrained problems that we saw in Table~\ref{table:1} exists and our plan of approaching the constrained problem using the list problem can be executed. 
Indeed, Ding and Xu~\cite{Ding_and_Xu_15} gave partition algorithms for a number of constrained problems. 
We make addition to their list which allows us to discuss new problems in this work. 
These additions and other discussions on approaching specific constrained problems using the list problem is discussed in Section~\ref{section: partition_algm} of Appendix. 
What we note here is that the approximation guarantee for the list problem carries over to all the constrained problem in Table~\ref{table:1}. 
We now look at our main results for the list $k$-service problem and its main implications for the constrained problems.

\subsection{Our Results}

We will show the following result for the list $k$-service problem.

\begin{theorem}[Main Theorem]\label{theorem:list_3_approx}
Let $0 < \veps \leq 1$.
Let $(C,L,k,d,\ell)$ be any $k$-service instance and let $\mathcal{C} = \{C_{1}, C_{2},\dotsc,C_{k}\}$ be any arbitrary clustering of the client set. 
There is an algorithm that, with probability at least $1/2$, outputs a list $\mathcal{L}$ of size $(k /\veps)^{O(k \, \ell^{\,2})}$, such that there is a $k$-center-set $S \in \mathcal{L}$ in the list such that
$\Psi (S,\C) \leq (3^{\ell} + \veps) \cdot \Psi^{*}(\C)$.
Moreover, the running time of the algorithm is $O \left( n \cdot (k /\veps)^{O(k \, \ell^{\,2})} \right)$.
For the special case when $C \subseteq L$, the algorithm gives a $(2^{\ell} + \veps)$-approximation guarantee.
\end{theorem}

\noindent Using the above Theorem together with Theorem~\ref{theorem: list_alpha_approx}, we obtain the following main results for the constrained $k$-means and $k$-median problems.
\begin{corollary}[$k$-means]
For any constrained version of the metric $k$-means problem with an efficient partition algorithm, there is a $(9+\veps)$-approximation algorithm with an FPT running time of $(k/\veps)^{O(k)} \cdot n^{O(1)}$. For a special case when $C \subseteq L$, the algorithm gives a $(4 + \veps)$-approximation guarantee.
\end{corollary}

\begin{corollary}[$k$-median]
For any constrained version of the metric $k$-median problem with an efficient partition algorithm, there is a $(3+\veps)$-approximation algorithm with an FPT running time of $(k/\veps)^{O(k)} \cdot n^{O(1)}$. For the special case when $C \subseteq L$, the algorithm gives a $(2 + \veps)$-approximation guarantee.
\end{corollary}

%Note that, a constrained problem does not need to have an efficient partition algorithm. 
Note that by Theorem~\ref{theorem: list_alpha_approx}, as long as the running time of the partition algorithm is $g(k) \cdot n^{O(1)}$, the total running time of the algorithm still stays FPT.
All the problems in Table~\ref{table:1} either have an efficient partition algorithm (polynomial in $n$ and $k$) or a partition algorithm with an FPT running time.
We discuss these partition algorithms in Section~\ref{section: partition_algm} of the Appendix.
Therefore, all the problems given in Table~\ref{table:1} admit a $(9+\veps)$-approximation and a $(3+\veps)$-approximation for the $k$-means and $k$-median objectives respectively.
It should be noted that other than the problems mentioned in Table~\ref{table:1}, our algorithm works for any problem that fits the framework of the constrained $k$-service problem (i.e., Definition~\ref{definition:constrained_k_service}) and has a partition algorithm.
This makes the approach extremely versatile since we one may be able to solve more problems that may arise in the future.\footnote{We note that new ways of modelling fairness in clustering is giving rise to new clustering problems with fairness constraints and some of these new problems may fit into this framework.}
The known results on constrained problems in Table~\ref{table:1} is summarised in Table~\ref{table:2}. 
Even though our work does not address the $k$-center problem, we state results on $k$-center  just to understand the state of art about these problems. 
Note that for {\bf all} these problems we obtain FPT time $(9+\veps)$-approximation and $(3+\veps)$-approximation for $k$-means and $k$-median respectively. For the special case when $C \subseteq L$ (a facility can be opened at any client location), we obtain FPT time $(4+\veps)$-approximation and $(2+\veps)$-approximation for $k$-means and $k$-median respectively.
There are some subtle differences in the problems in Table~\ref{table:1} and Table~\ref{table:2}. This is to be able to compare our results with known results. We will highlight these differences in the related work section.

\begin{table}[h]
\begin{adjustbox}{width=\columnwidth,center}
\centering
\setcellgapes{1ex}\makegapedcells
\begin{tabular}{|c|l|c|c|c|}
\hline
\# & {\bf Problem} & {\bf Metric $k$-center} & {\bf Metric $k$-median} & {\bf Metric $k$-means}\\ \hline
1. & \makecell[l]{$r$-gather $k$-service \\ (uniform case) \\ } & \makecell[l]{$2$-approx.~\cite{rgather:k_center_2010_Aggarwal}
} & \makecell[l]{ 
7.2-approx~\cite{rgather:k_all_2018_Ding} \\ (for $C = L$) \\(in FPT time)} & \makecell[l]{ 86.9-approx~\cite{rgather:k_all_2018_Ding} \\ (for $C = L$) \\(in FPT time)
} \\ \hline
2. & \makecell[l]{$r$-Capacity $k$-service \\ (uniform case)} & \makecell[l]{$6$-approx.~\cite{capacitated:kcenter_2000_khuller}
} & \makecell[l]{ $(3+\veps)$-approx~\cite{capacitated:FPT_2019_vincent} \\ (in FPT time)} & \makecell[l]{ $(9+\veps)$-approx~\cite{capacitated:FPT_2019_vincent} \\ (in FPT time)} \\ \hline
3. & \makecell[l]{$l$-Diversity $k$-service} & \makecell[c]{$(2+\veps)$-approx.~\cite{L_diversity:2010_kcenter_Li_Jian}
} & \makecell[c]{ - } & \makecell[c]{ - } \\ \hline
4. & \makecell[l]{Chromatic $k$-service} & \makecell[c]{ - } & \makecell[c]{ - } & \makecell[c]{ - } \\ \hline
5. & \makecell[l]{Fault tolerant $k$-service} & \makecell[c]{$3$-approx. ~\cite{fault:kcenter_2000_khuller}
} & \makecell[c]{ 93-approx.~\cite{fault:kmedian_2014_non_uniform_haji_li_SODA} } & \makecell[c]{ - } \\ \hline
6. & \makecell[l]{Semi-supervised $k$-service} & \makecell[c]{-} & \makecell[c]{ - } & \makecell[c]{ - } \\ \hline
7. & \makecell[l]{Uncertain $k$-service \\ (assigned version) \\}  & \makecell[c]{$10$-approx.~\cite{uncertain:kcenter_2018_Alipour}
} & \makecell[l]{ $(6.35+\veps)$-approx.~\cite{uncertain:kcenter_2008_Cormode}\vspace*{0.5mm} \\
(for $C \subseteq L$)} & \makecell[l]{ $(74+\veps)$-approx.~\cite{uncertain:kcenter_2008_Cormode} \vspace*{0.5mm}\\
(for $C \subseteq L$)} \\ \hline
8. & \makecell[l]{Outlier $k$-service} & \makecell[c]{3-approx.~\cite{outlier:kcenter_2001_Charikar}} & \makecell[c]{ $(7+\veps)$-approx.~\cite{outlier:kmeans_2018_Ravishankar} } & \makecell[c]{ $(53+\veps)$-approx.~\cite{outlier:kmeans_2018_Ravishankar} } \\ \hline
 \multicolumn{5}{l}{} \\[-1.6em] \hline 
\multicolumn{5}{|l|}{\begin{tabular}[c]{@{}l@{}}For the Euclidean $k$-means and $k$-median (where $C \subseteq L = \mathbb{R}^{d}$), all the constrained problems \vspace*{-2mm} \\[-0.6em] have a FPT time $(1+\veps)$  approximation algorithm~\cite{Ding_and_Xu_15,bjk18}.\protect %\footnotemark
\end{tabular}}
 \\ \hline
\end{tabular}
\end{adjustbox}
\vspace*{0.5mm}
\caption{Known results for the constrained clustering problems. Note that for {\bf all} the above problems we obtain FPT time $\mathbf{(3+\veps)}$-approximation and $\mathbf{(9+\veps)}$-approximation for $k$-median and $k$-means respectively. For the special case when $C \subseteq L$ (a facility can be opened at any client location), we obtain FPT time $\mathbf{(2+\veps)}$-approximation and $\mathbf{(4+\veps)}$-approximation for $k$-median and $k$-means respectively.
}
\label{table:2}
\end{table}
%\footnotetext[3]{For the Euclidean outlier $k$-service problem the $(1+\veps)$-approximation is for fixed $k$ and $m$}

Moreover, we can convert our algorithms to streaming algorithms using the technique of Goyal~\emph{et al.}~\cite{gjk19}. 
We basically require a streaming version of our algorithm for the list $k$-service problem and a streaming partition algorithm for the constrained $k$-service problem.
In Section~\ref{section:streaming}, we will design a constant-pass log-space streaming algorithm for the list $k$-service problem. 
We already know streaming partition algorithms for the various constrained $k$-service problems~\cite{gjk19}. 
This would give a streaming algorithm for all the problems given in Table~\ref{table:1} except for the $\ell$-diversity and chromatic $k$-service problems.  
Although \emph{single}-pass streaming algorithms are considered much useful, it is interesting to know that there is a constant-pass streaming algorithm for many constrained versions of the $k$-service problem.

\subsection{Related Work}
A unified framework for constrained $k$-means/$k$-median problems was introduced by Ding and Xu~\cite{Ding_and_Xu_15}.
Using this framework, they designed a PTAS (fixed $k$) for various constrained clustering problems. 
However, their study was limited to the Euclidean space where $C \subseteq L = \mathbb{R}^{d}$. 
Their results were obtained through an algorithm for the list version of the $k$-means problem (even though it was not formally defined in their work). 
The running time of this algorithm was $O(nd \cdot (\log n)^{k} \cdot 2^{poly(k/\veps)})$ and the list size was $(\log n)^{k} \cdot 2^{poly(k/\veps)}$.
Bhattacharya \emph{et al.}~\cite{bjk18} formally defined and studied the list $k$-service problem. 
They obtained a faster algorithm for the list problem with running time to $O(nd \cdot (k/\veps)^{O(\log(k/\veps))})$ and list size to $ (k/\veps)^{O(\log(k/\veps))}$ for the constrained $k$-means/$k$-median problem. 
Recently, Goyal \emph{et al,}~\cite{gjk19} designed streaming algorithms for various constrained versions of the problem by extending the previous work of Bhattacharya \emph{et al.}~\cite{bjk18}. 
In this paper, we study the problem in general metric spaces while treating $L$ and $C$ as separate sets. 
More importantly, we design an algorithm that gives a better approximation guarantee than the previously known algorithms by taking advantage of FPT running time. 
Moreover, for many problems, it is the first algorithm that achieves a constant-approximation in FPT running time. 
Please see Table~\ref{table:2} for the known results on the problem.
We have a detailed discussion on these problems in Section~\ref{sec:constrained_problems} of the Appendix.

In the introduction, we would specifically like to discuss the result of Addad \emph{et al.}~\cite{capacitated:FPT_2019_vincent} for the capacitated $k$-service problem. 
Their definition of the capacitated $k$-service problem is different from the one mentioned in Table~\ref{table:1} that we are consider. 
Following is their definition of the capacitated $k$-service problem.

\begin{definition}[Addad {\em et al.}~\cite{capacitated:FPT_2019_vincent}]
Given an instance $\mathcal{I} = (C,L,k,d,\ell)$ of the $k$-service problem and a capacity function $r : L \to \mathbb{Z}_{+}$, find a set $F \subseteq L$ of $k$ facilities such that the assignment cost $\sum_{j \in C} \min_{i \in F} d^{\ell}(j,i)$ is minimized, and no more than $r_{i}$ clients are assigned to a facility $i \in L$. 
\end{definition}
Note that in the above definition, a facility has a capacity of $r_{i}$ whereas in our definition a cluster has a capacity of $r_{i}$. 
This is an important difference since in our case we can assign a cluster of size $r_{i}$ to any facility location which is not possible by their definition. 
However, for the uniform capacities the problem definitions are equivalent and the results become comparable.
We match the approximation guarantees obtained Addad {\em et al.}~\cite{capacitated:FPT_2019_vincent} for the uniform case even though using very different techniques.

As we mentioned earlier, the unconstrained metric $k$-median problem is hard to approximate within a factor of $(1+2/e)$, and the metric $k$-means problem is hard to approximate within a factor of $(1+8/e)$. 
Surprisingly this lower bound persists even if we allow an FPT running time~\cite{vincent_hardness_2019}. 
However, this FPT lower bound is based on some recent complexity theoretic conjecture. 
The problem also has a matching upper bound algorithm with an FPT running time~\cite{vincent_hardness_2019}.  
So, the unconstrained $k$-means and $k$-median problems in the metric setting is fairly well understood.
On the other hand, our understanding of most constrained versions of the problem is still far from complete. 
We believe that our work is be an important step in understanding constrained problems in general metric spaces.

%%%%%% Our Techniques %%%%%%%%%%

\subsection{Our Techniques}
In this section, we discuss our sampling based algorithm for list $k$-service problem.
As described earlier, an FPT algorithm for the list $k$-service problem gives an FPT algorithm for a constrained version of the $k$-service problem that has an efficient or FPT-time partition algorithm. 
Our sampling based algorithm is similar to the algorithm of Goyal \emph{et al.}~\cite{gjk19} that was specifically designed for the Euclidean setting.
However, working in a metric space instead of Euclidean space poses challenges as some of the main tools used for analysis in the Euclidean setting cannot be used in metric spaces.
We carefully devise and prove new sampling lemmas that makes the high-level analysis of Goyal {\em et al.}~\cite{gjk19} go through.
Our algorithm is based on $D^\ell$-sampling. Given a point set $F$, $D^\ell$-sampling a point from the client set $C$ w.r.t. center set $F$ means sampling using the distribution where the sampling probability of a client $x \in C$ is $\frac{\Phi(F, \{x\})}{\Phi(F, C)} = \frac{\min_{f \in F}d^{\ell}(f, x)}{\sum_{y \in C} \min_{f \in F}d^{\ell}(f, y)}$. In case $F$ is empty, then $D^\ell$-sampling is the same as uniform sampling.
Following is our algorithm for the list $k$-service problem:

\begin{Algorithm}[h]
\begin{framed}
\LK ~($C, L, k, d, \ell, \veps$) \vspace{1mm}\\
\hspace*{0.3in} {\bf Inputs}: $k$-service instance $(C,L,k,d,\ell)$ and accuracy $\veps$ \\
\hspace*{0.3in} {\bf Output}: A list $\mathcal{L}$, each element in $\mathcal{L}$ being a $k$-center set \vspace*{2mm}\\
\hspace*{0.3in} {\bf Constants}: $\beta = 4^{\ell-1} \cdot \left( \myfrac[2pt]{\ell^{\ell} \cdot 3^{\ell^{2}+4\ell+3}}{\veps^{ \, \ell+1}}+1 \right)$; $\gamma = \myfrac[2pt]{
\ell^{\ell} \cdot 3^{\ell^{2} + 5\ell+1}}{\veps^{\, \ell}}$; $\mathbf{\eta} = \myfrac[2pt]{\alpha \, \beta \, \gamma \, k \cdot 3^{\ell+2}}{\veps^{2}}$\\[4pt]
\hspace*{0.1in} (1) \ \ \ Run any $\alpha$-approximation algorithm for the \emph{unconstrained} $k$-service  \\
\hspace*{0.1in} \ \ \ \ \ \ \ \ instance $(C,C,k,d,\ell)$ and let $F$ be the obtained center-set. \\
\hspace*{0.6in} ({\it $k$-means++ ~\cite{kmeanspp} is one such algorithm.})\\
\hspace*{0.1in} (2) \ \ \ $\mathcal{L} \gets \emptyset$\\
\hspace*{0.1in} (3) \ \ \ Repeat $2^k$ times:\\
\hspace*{0.1in} (4)\hspace*{0.3in}  \ \ \ Sample a multi-set $M$ of $\eta k$ points from $C$ using $D^{\ell}$-sampling w.r.t. \\
\hspace*{0.1in} \hspace*{0.5in} \ \ \ center set $F$\\
\hspace*{0.1in} (5)\hspace*{0.3in}  \ \ \ $M \gets M \cup F$ \\
\hspace*{0.1in} (6)\hspace*{0.3in}  \ \ \ $T \gets \emptyset$ \\
\hspace*{0.1in} (7)\hspace*{0.3in} \ \ \ For every point $x$ in $M$:\\
\hspace*{0.1in} (8)\hspace*{0.9in} 
$T \gets T \cup \{k \text{ points in $L$ that are closest to $x$}\}$\\
\hspace*{0.1in} (9)\hspace*{0.3in} \ \ \ For all subsets $S$ of $T$ of size $k$:\\
\hspace*{0.1in} (10)\hspace*{0.84in} $\mathcal{L} \gets \mathcal{L} \cup \{ S\}$\\
\hspace*{0.1in} (11) \ \  return($\mathcal{L}$)
\end{framed}
\vspace*{-4mm}
\caption{Algorithm for the list $k$-service problem}
\label{algorithm:kmeans_intro}
\end{Algorithm}
Let us discuss some of the main ideas of the algorithm and its analysis.
Note that in the first step, we obtain a center-set $F \subseteq C$ which is an $\alpha$-approximation for the \emph{unconstrained} $k$-service instance $(C,C,k,d,\ell)$.
That is, $\Phi(F,C) \leq \alpha \cdot OPT(C,C)$. 
One such algorithm is the $k$-means++ algorithm~\cite{kmeanspp} that gives an $O(4^{\ell} \cdot \log k)$-approximation guarantee and a running time $O(nk)$.
Now, let us see how the center-set $F$ can help us.
Let us focus on any cluster $C_{i}$ of a target clustering $\C = \{C_{1}, \dotsc, C_{k}\}$. 
We note that the closest facility to a uniformly sampled client from any any client set $C_i$ provides a constant approximation to the optimal $1$-median/means cost for $C_i$ in expectation.
This is formalized in the next lemma.
This lemma (or a similar version) has been used in multiple other works in analysing sampling based algorithms.
\begin{lemma}~\label{lemma:exp_intro}
Let $S \subseteq C$ be any subset of clients and let $f^{*}$ be any center in $L$. 
If we uniformly sample a point $x$ in $S$ and open a facility at the closest location in $L$, then the following identity holds:
\[
\mathbb{E}[\Phi(t(x),S)] \leq 3^{\ell} \cdot \Phi(f^{*},S), \textrm{ where $t(x)$ is the closest facility location from $x$.}
\]
\end{lemma}
Unfortunately, we cannot uniformly sample from $C_i$ directly since $C_{i}$ is not known to us.
Given this, our main objective should be to use $F$ to try to uniformly sample from $C_{i}$ so that we could achieve a constant approximation for $C_{i}$.
Let us do a case analysis based on the distance of points in $C_i$ from the nearest point in $F$.
Consider the following two possibilities:
The first possibility is that the points in $C_{i}$ are close to $F$. 
If this is the case, we can uniformly sample a point from $F$ instead of $C_{i}$. 
This would incur some extra cost. 
However, the cost is small and can be bounded. 
To cover this first possibility, the algorithm adds the entire set $F$ to the set of sampled points $M$ (see line (5) of the algorithm).
The second possibility is that the points in $C_{i}$ are far-away from $F$. 
In this case, we can $D^{\ell}$-sample the points from $C$.
Since the points in $C_{i}$ are far away, the sampled set would contain a good portion of points from $C_{i}$ and the points will be {\em almost} uniformly distributed. 
We will show that almost uniform sampling is sufficient to apply Lemma~\ref{lemma:exp_intro} on $C_{i}$. 
However, we would have to sample a large number of points to boost the success probability. 
This requirement is taken care of by line (4) of the algorithm.
Note that we may need to use a hybrid approach for analysis since the real case may be a combination of the first and second possibility.
Most of the ingenuity of this work lies in formulating and proving appropriate sampling lemmas to make this hybrid analysis work.

To apply lemma~\ref{lemma:exp_intro}, we need to fulfill one more condition. 
We need the closest facility location from a sampled point. This requirement is handled by lines (7) and (8) of the algorithm. However, note that the algorithm picks $k$-closest facility locations instead of just one facility location. We will show that this step is crucial to obtain a \emph{hard-assignment} solution for the problem. 
Finally, the algorithm adds all the potential center sets to a list $\mathcal{L}$ (see line (9) and (10) of the algorithm). The algorithm repeats this procedure $2^{k}$ times to boost the success probability (see line (3) of the algorithm).
We will show the following result from which our main theorem (Theorem~\ref{theorem:list_3_approx}) trivially follows.

\begin{theorem} \label{theorem:list_k_service_intro}
Let $0 < \veps \leq 1$.
Let $(C,L,k,d,\ell)$ be any $k$-service instance and let $\mathcal{C} = \{C_{1}, C_{2},\dotsc,C_{k}\}$ be any arbitrary clustering of the client set.
The algorithm \\
\LK($C,L,k,d,\ell,\veps$), with probability at least $1/2$, outputs a list $\mathcal{L}$ of size $(k /\veps)^{O(k \, \ell^{\,2})}$, such that there is a $k$ center set $S \in \mathcal{L}$ in the list such that
\[
    \Psi (S,\C) \leq (3^{\ell} + \veps) \cdot \Psi^{*}(\C).
\]
Moreover, the running time of the algorithm is $O \left( n \cdot (k /\veps)^{O(k \, \ell^{\,2})} \right)$. For the special case of $C \subseteq L$, the approximation guarantee is $(2^{\ell}+\veps)$.
\end{theorem}
The details of the analysis is given in Appendix~\ref{section:list_k_service}.

\subsection{A Matching Lower Bound on approximation}
We gave sampling based algorithms and showed an approximation guarantee of $(3^{\ell} + \veps)$ (and $(2^\ell + \veps)$ for the special case $C \subseteq L$).
In this subsection, we show that our analysis of the approximation factor is tight.
More specifically, we will show that our algorithm does not provide better than ($3^{\ell}-\delta'$) approximation guarantee for arbitrarily small $\delta' > 0$ (and $2^{\ell} - \delta'$ for the case $C \subseteq L$). 
To show this, we create a {\em bad instance} for the problem in the following manner. We create  the instance using an undirected weighted graph where $C \cup L$ is the vertex set of the graph and the shortest weighted path between two vertices defines the distance metric. The set $C$ is partitioned into the sub-sets $C_{1}, C_{2}, \dotsc, C_{k}$, and $L$ is partitioned into the sub-sets $L_{1}, L_{2}, \dotsc, L_{k}$.
The sub-graphs over $C_{1} \cup L_{1}, C_{2} \cup L_{2},\dotsc$, and  $C_{k} \cup L_{k}$ are all identical to each other. 
Let us describe the sub-graph over vertex set $C_{i} \cup L_{i}$ in general. 
In this sub-graph, all the clients are connected to a common facility location $f_{i}^{*}$ with an edge of unit weight. 
Also, every client is connected to a distinct set of $k$ facility locations with an edge of weight $(1-\delta)$. 
We denote this set by $T(x)$ for a client $x \in C_{i}$.
Figure~\ref{fig:undirected_graph} shows the complete description of this sub-graph. 
Lastly, all pairs of sub-graphs $C_{i} \cup L_{i}$ and $C_{j} \cup L_{j}$ are connected with an edge $(f_{i}^{*},f_{j}^{*})$ of weight $\Delta \gg |C|$. 
This completes the construction of the bad instance.

Let us define a target clustering on the instance. Consider the unconstrained $k$-service problem. It is easy to see that $\C = \{ C_{1}, C_{2}, \dotsc,C_{k} \}$ is an optimal clustering for this instance.
The optimal cost of a cluster $C_{i}$ is $\Phi(f_{i}^{*},C_{i}) = |C_{i}|$, and the optimal cost of the entire instance is $OPT = \sum_{i} |C_{i}| = |C|$.

\begin{figure}[h]
    \centering
    \includegraphics[scale=0.8]{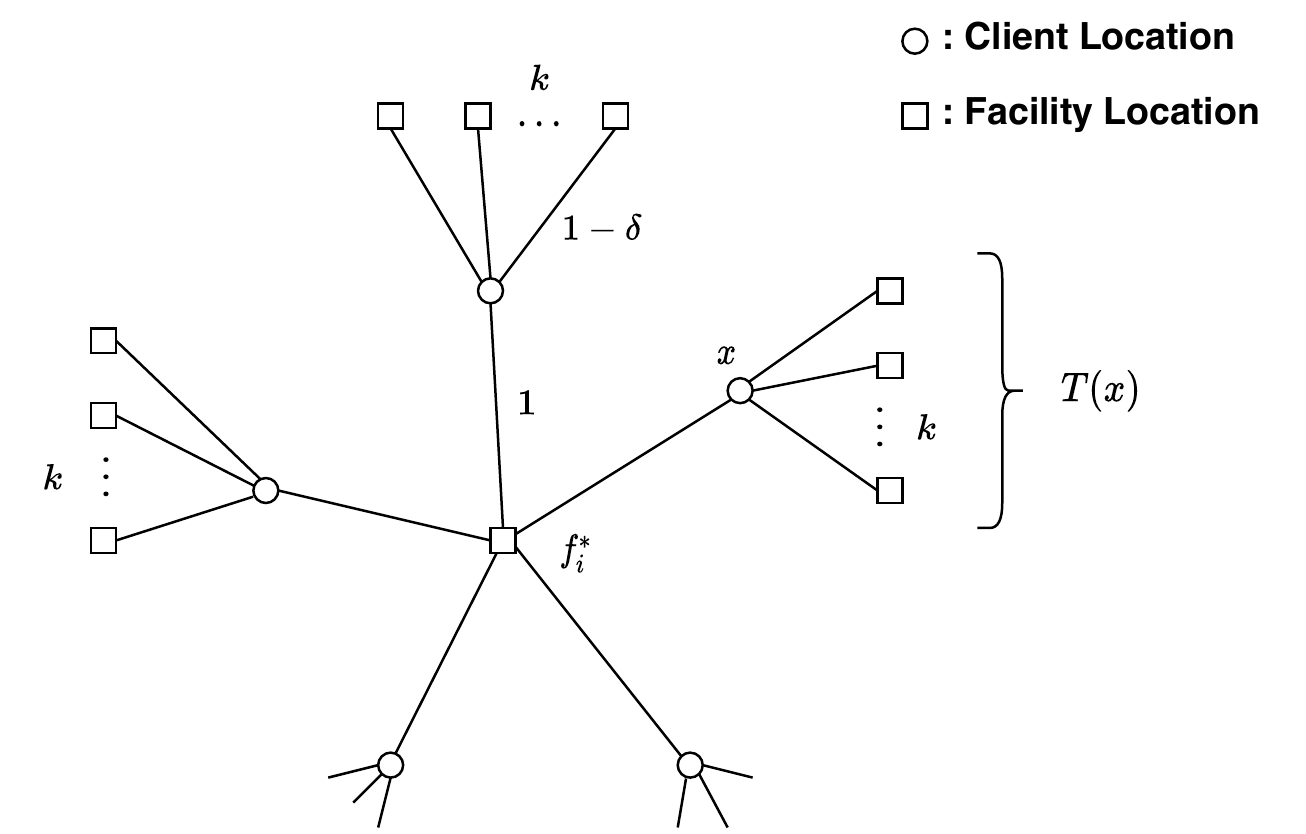}
    \vspace*{1mm}
    \caption{An undirected weighted sub-graph on $C_{i} \cup L_{i}$.}
    \label{fig:undirected_graph}
\end{figure}

Now, we will show that any list $\mathcal{L}$ produced by the algorithm \LK ~does not contain any center-set that can provide better than $(3^{\ell} - \delta')$-approximation for $\C$. 
To show this, let us examine every center-set in the list $\mathcal{L}$ produced by \LK.
Note that the set $T$ obtained in line (10) of the algorithm does not contain any optimal facility location $f_{i}^{*}$ because $f_{i}^{*}$ does not belong to $T(x)$. 
Therefore, no center set in the list contains any of the optimal facility locations $\{f_1^*, ..., f_k^*\}$. 
Let us evaluate the clustering cost corresponding to every center set in the list. 
Let $F = \{f_{1},f_{2},\dotsc,f_{k}\}$ be a center-set in the list. We have two possibilities for the facilities in $F$. The first possibility is that, there are at least two facilities in $F$, that belongs to the same sub-graph $C_{i} \cup L_{i}$. In this case, the cost of the target clustering is  $\Psi(F,\C) > \Delta \gg OPT$.  
So in this case, $F$ gives an unbounded clustering cost. 
Let us consider the second possibility that all facilities in $F$ belong to different sub-graphs. 
Without loss of generality, we can assume that $f_{i} \in L_{i}$. 
Since $f_{i}$ can not be the optimal facility location, we can further assume that $f_{i} \in T(x)$ for some $x \in C_{i}$. The cost of a cluster in this case is $\Phi(f_{i},C_{i}) = (3-\delta)^{\ell}(|C_{i}|-1) + (1-\delta)^{\ell} > (3-\delta)^{\ell}(|C_{i}|-1)$ . Hence, the overall cost of the instance is $\Psi(F,\C) > (3-\delta)^{\ell} \cdot ( |C| - k) \geq (3-\delta)^{\ell} \cdot |C| - 3^{\ell} \, k \geq (3^{\ell}-\delta') \cdot |C|$, for $\delta' = 3^{\ell-1} \cdot \ell \,\delta + \fdfrac{3^{\ell}k}{ |C|}$. Therefore, we can say that list does not contain any center set that can provide  better than ($3^{\ell} - \delta'$) approximation guarantee for $\C$.

\begin{theorem}
For any $0 <\delta' \leq 1$, there are instances of the $k$-service problem for which the algorithm \LK$(C,L,k,d,\ell,\veps)$ does not provide better than $(3^{\ell}-\delta')$ approximation guarantee.
\end{theorem}

\noindent Now, let us examine the same bad instance when we have the flexibility to open a facility at a client location. 
In this case, we have a third possibility that $F = \{ f_{1}, f_{2}, \dotsc, f_{k} \}$  such that $f_{i}$ is some client location in $C_{i}$. The cost of a cluster in this case is $\Phi(f_{i},C_{i}) = 2^{\ell} \cdot (|C_{i}|-1)$  and the overall cost the instance is $\Psi(F,\C) = 2^{\ell}  \cdot |C| - 2^{\ell} \cdot k = (2^{\ell} - \delta' ) \cdot |C|$, for $\delta' = 2^{\ell} \cdot k/|C|$. 
So for the special case $C \subseteq L$, we obtain the following theorem.

\begin{theorem}
For any $0 <\delta' \leq 1$, there are instances of the $k$-service problem (with $C \subseteq L$), for which the algorithm~\LK$(C,L,k,d,\ell,\veps)$ does not provide better than $(2^{\ell}-\delta')$ approximation guarantee.
\end{theorem}

\subsection{Streaming Algorithms}
\label{section:streaming}
In this subsection, we discuss how to obtain constant-pass streaming algorithm using the ideas of Goyal \emph{et al.}~\cite{gjk19}.
Our offline algorithm has two main components, namely: the list $k$-service algorithm and partition algorithm. The list $k$-service procedure is common to all constrained versions of the problem. 
However, the partition algorithm differs for different constrained versions.
First, let us convert \LK($C,L,k,d,\ell$) algorithm to a streaming algorithm.
\begin{framed}
\begin{enumerate}
    \item In the first pass, we run a streaming $\alpha$-approximation algorithm for the instance $(C,C,k,d,\ell)$. For this, we can use the streaming algorithm of Braverman {\it{et al.}}~\cite{bmo11}. The algorithm gives a constant-approximation with the space complexity of $O(k \log n)$.
    \item In the second pass, we perform the $D^{\ell}$-sampling step using the \emph{reservoir sampling} technique~\cite{vitter85}.
    \item In the third pass, we find the $k$-closest facility locations for every point in $M$. 
\end{enumerate}
\end{framed}

\noindent This gives us the following result.
\begin{theorem}
There is a 3-pass streaming algorithm for the list $k$-service problem, with the running time of $O(n \cdot f(k,\veps))$ and space complexity of $f(k,\veps) \cdot \log n$, where $f(k,\veps) = (k/\veps)^{O(k \ell^{2})}$.
\end{theorem}

\noindent Now, let us discuss the partition algorithms in streaming setting. For the $l$-diversity and chromatic $k$-service problems, it is known that there is no deterministic log-space streaming algorithm~\cite{gjk19}.
For the remaining constrained problems, there are streaming partition algorithms that are discussed in~\cite{gjk19} (for the Euclidean setting) and Section~\ref{section: partition_algm} of the Appendix (metric setting). 
Note that all of the streaming partitioning algorithms do not give an optimal partitioning but only a partitioning that is close to the optimal.
%Rather, they produce a partitioning of cost at most $(1+\veps)$ times the optimal.
Each algorithm makes at most $3$-pass over the data-set and takes logarithmic space complexity. 
The partition algorithm, together with the list $k$-service algorithm, gives the following main results.
\begin{theorem}
For the following constrained $k$-service problems there is a $6$-pass streaming algorithm that gives a $(3^{\ell} + \veps)$-approximation guarantee.
\begin{enumerate}
    \item $r$-gather $k$-service problem
    \item $r$-capacity $k$-service problem
    \item Fault-tolerant $k$-service problem
    \item Semi-supervised $k$-service problem
    \item Uncertain $k$-service problem (assigned case)
\end{enumerate}
The algorithm has the space complexity of $O(f(k,\veps,\ell) \cdot \log n)$ and the running time of $O(f(k,\veps,\ell) \cdot n^{O(1)})$, where $f(k,\veps,\ell) = (k/\veps)^{O(k \ell^{2})}$. Further, the algorithm gives $(2^{\ell}+\veps)$-approximation guarantee when $C \subseteq L$.
\end{theorem}

\begin{theorem}
For the outlier $k$-service problem there is a 5-pass streaming algorithm that gives a $(3^{\ell} + \veps)$-approximation guarantee. The algorithm has space complexity of $O(f(k,m,\veps,\ell) \cdot \log n)$ and running time of $f(k,m,\veps,\ell) \cdot n^{O(1)}$, where $f(k,m,\veps,\ell) = ((k+m)/\veps)^{O(k \cdot \ell^{2})}$. Further, the algorithm gives $(2^{\ell}+\veps)$-approximation guarantee when $C \subseteq L$.
\end{theorem}

\section{Conclusion and Open Problems}\label{section:conclusion}

In this paper, we worked within the unified framework of Ding and Xu~\cite{Ding_and_Xu_15} to obtain simple sampling based algorithms for a range of constrained $k$-median/means problems in general metric spaces. Surprisingly, even working within this high-level framework, we obtained better (or matched) approximation guarantees of known results that were designed specifically for the constrained problem. 
On one hand, this shows the versatility of the unified approach along with the sampling method. 
On the other hand, it encourages us to try to design algorithms with better approximation guarantees for these constrained problems. 
Our matching approximation lower bound for the sampling algorithm suggests that further improvement may not be possible through sampling based ideas.
On the lower bound side, it may be useful to obtain results similar to that for the unconstrained setting where approximation lower bounds of $(1+2/e)$ and $(1+8/e)$ are known for $k$-median and $k$-means respectively~\cite{vincent_hardness_2019}. 
Another direction is to find other constrained problems that can fit into the unified framework and can benefit from the results in this work.

\section*{Acknowledgements} The authors would like to thank Anup Bhattacharya for useful discussions.

%%
%% Bibliography
%%

%% Please use bibtex, 
\addcontentsline{toc}{section}{References}
\bibliographystyle{alpha}
\bibliography{references}

\appendix
\section{Appendix}
To keep the high-level discussion concise, we gave a general introduction in the main paper highlighting our main results.
We will use the space in the appendix to give the detailed paper. 
This is done section wise starting with a detailed discussion on the constrained problems followed by technical sections starting with preliminaries and then a section for each of the technical contributions discussed in the main paper.

\section{Discussion on Constrained Problems}\label{sec:constrained_problems}
In this section, we will look at each of the constrained problems in Table~\ref{table:1} in more detail.

\subsection{$r$-Gather $k$-Service Problem}\label{subsection:B1}
As you can see from Table~\ref{table:1}, the problem is defined with a lower bound constraint on the cluster sizes, i.e., each cluster $C_{i}$ must contain at least $r_{i}$ clients.
A problem similar to the above can be defined where the lower bound constraint is on the facilities instead of clusters.
That is, a facility $i \in L$ must serve at least $r_{i}$ clients.
Note that above two problems are not the same.
However, these problem definitions become equivalent in the \emph{uniform} setting where lower bound is the same (i.e., $r_1=r_2=...=r$). 
The second problem arises in the context of the facility location (where there is no bound on facilities but there is a facility opening cost) called the \emph{lower bound facility location} (LBLF) problem where the lower bound constraint is on the facilities.
The first problem in the uniform setting is called the $r$-gather problem and we will refer to the version where $r_i$'s may not be the same as the non-uniform $r$-gather problem.

The lower bound facility location problem was first introduced by Guha~\emph{et al.}~\cite{rgather:facility_location_2000_guha_FOCS} and Karger~\emph{et al.}~\cite{rgather:facility_location_karget_minkoff_FOCS} to solve various \emph{network design} problems. Both of these works gave a bi-criteria approximation algorithm for the problem, where the lower bound constraint is violated by a constant factor and the solution cost is at most constant times the optimal. Svitkina~\cite{rgather:facility_location_2010_zoya} gave the first constant-approximation algorithm for the \emph{uniform} LBLF problem. The approximation guarantee of the algorithm was 448, which was later improved to 82.6 by Ahmadian and Swamy~\cite{rgather:facility_location_2011_Ahmadian}.
Recently, Shi Li~\cite{rgather:facility_location_2019_ShiLi} gave a 4000-approximation algorithm for the \emph{non-uniform} lower bound facility location problem. 
     
For the (uniform) $r$-gather $k$-\emph{median} problem, the algorithms of~\cite{rgather:facility_location_2010_zoya} and~\cite{rgather:facility_location_2011_Ahmadian} can be adapted to obtain an $O(1)$-approximation guarantee~\cite{rgather:Misc_2016_Ahmadian}.
Ding~\cite{rgather:k_all_2018_Ding} gave a \emph{FPT} $(3 \lambda +2)$-approximation and $(18 \lambda+16)$-approximation algorithm for the uniform $r$-gather $k$-median and $k$-means problem under the assumption $C = L$, respectively. 
Here, $\lambda$ denotes the approximation guarantee of any \emph{unconstrained} $k$-median or $k$-means algorithm. 
We can use the FPT algorithm of Addad~\emph{et al.}~\cite{vincent_hardness_2019} for the unconstrained problems that has $\lambda = 1+2/e$ and $\lambda = 1+8/e$ for the $k$-median and $k$-means problem respectively. 
Moreover, these bounds are tight conditioned on some recent complexity theory conjectures. 
So, the algorithm of Ding~\cite{rgather:k_all_2018_Ding} gives a $7.2$ and $86.9$-approximation for the $r$-gather $k$-median and $k$-means problems respectively. 
Even though this algorithm is FPT in $k$, its advantage over the previous algorithms is that it considers both the lower and upper bounds on the size of the clusters. 
Under the assumption that $C \subseteq L$ ($C = L$ is a special case of $C \subseteq L$), we design a $2$ and $4$ approximation algorithm for the \emph{non-uniform} $r$-gather problem that has an FPT running time. 
This is an improvement over the result of Ding~\cite{rgather:k_all_2018_Ding}. 
Moreover, our algorithm also considers the lower and upper bounds on the size of clusters.

The $r$-gather $k$-\emph{center} problem is widely used in privacy-preserving data publication. The problem was first introduced by Aggarwal~\emph{et al.} \cite{rgather:k_center_2010_Aggarwal} and is based on the famous $k$-\emph{anonymity} principle of Sweeney~\cite{rgather:k_anonymity_2002_Sweeney}. A $2$-approximation algorithm is known for the uniform $r$-gather $k$-center problem in offline setting (Section 2.4 of ~\cite{rgather:k_center_2010_Aggarwal}). This is the best approximation possible for the problem since the $k$-center problem (for $r = 1$) does not admit better than ($2-\veps$)-approximation for any $\veps > 0$, assuming $\mathsf{P} \neq \mathsf{NP}$. Moreover, a $6$-approximation is known for the uniform $r$-gather $k$-center problem in the streaming setting~\cite{rgather:k_center_streaming}.

\subsection{$r$-Capacity $k$-Service Problem}    
As you can see from Table~\ref{table:1}, the problem is defined by an upper bound constraint on the size of the clusters. 
This constraint is useful in scenarios where a facility can serve only a certain number of clients due to limited resources. 
A problem similar to the above can be defined where capacities are imposed on individual facilities instead of clusters. 
Note that the two problems are different.
However, in case of \emph{uniform} capacities (i.e., $r_1=r_2=...=r$), both problem definitions are equivalent.
The first problem in the uniform setting is called the $r$-capacity problem and we will refer to the version where $r_i$'s may not be the same as the {\em non-uniform} $r$-capacity problem. The second problem is called the {\em capacitated $k$-median/means} problem.
The second problem arises in the context of {\em facility location} where there is no bound on the number of open facilities but there is facility opening cost.
The facility location problem with upper bounds in capacities is known as the \emph{capacitated facility location} (CFL) problem.

The capacitated facility location problem has been studied extensively \cite{capacitated:facility_location_2001_pal,capacitated:facility_location_2003_Mahadian_Pal,capacitated:facility_location_2005_Zhang,capacitated:facility_location_2012_naveen,capacitated:facility_location_uniform_1999_chudak,capacitated:facility_location_uniform_2000_Madhukar,capacitated:facility_location_uniform_2010_Naveen}. 
The best known approximation guarantee for the problem is $5$ in the non-uniform setting~\cite{capacitated:facility_location_2012_naveen} and $3$ in the uniform setting~\cite{capacitated:facility_location_uniform_2010_Naveen}. 
All these approximation algorithms are based on the local-search technique. 
The LP based constant approximation algorithms are also known for the problem~\cite{capacitated:facility_location_LP_based_2004_uniform_levi,capacitated:facility_location_LP_based_2014_An_Mohit}. 
    
For the capacitated $k$-median/$k$-means problem, no constant-factor approximation is known even in the uniform setting. 
However, various bi-criteria approximation algorithms are known for the problem~\cite{capacitated:kmedian_2017_Li_uniform,capacitated:kmedian_2017_Li,capacitated:kmedian_2017_Byrka_uniform,capacitated:kmedian_2017_Demerci_Li}, which either violate the capacity or {\em cardinality} constraint (the constraint on the number of open facilities) by a constant factor. 
The problem has been also studied from the perspective of fixed parameter tractability with $k$ as the parameter~\cite{capacitated:FPT_2018_Xu,capacitated:FPT_2018_Byrka,capacitated:FPT_2019_vincent}. The algorithm of Addad~\emph{et al.}~\cite{capacitated:FPT_2019_vincent} is based on the coreset technique, and gives a $3$ and $9$-approximation for the $k$-median and $k$-means objectives respectively, in \emph{FPT} time. 
Our algorithm gives the same bounds as that of Addad {\em et al.}~\cite{capacitated:FPT_2019_vincent} but for the non-uniform $r$-capacity problem. 
Note that the uniform version is a special case of the non-uniform version and the uniform $r$-capacity problem is equivalent to the uniform capacitated $k$-median/means problem.  
This means that our algorithm is also $3$ and $9$ approximation algorithm for the uniform capacitated $k$-median and $k$-means problems respectively.
So, in the uniform setting, we match the result of Addad {\em et al.}~\cite{capacitated:FPT_2019_vincent}.
Moreover, our algorithm can be converted to a streaming algorithm. 
    
The capacitated $k$-\emph{center} problem has also been well studied~\cite{capacitated:kcenter_1992_Barillan,capacitated:kcenter_2000_khuller,capacitated:kcenter_2012_khuller,capacitated:kcenter_2015_An}. 
The best approximation guarantee for the problem is $9$~\cite{capacitated:kcenter_2015_An} in the non-uniform setting and $6$ in the uniform setting~\cite{capacitated:kcenter_2000_khuller}. Some interesting open problem for the capacitated problems are discussed in the survey of An and Svensson~\cite{Survey_2017_An_Ola_Svensson}.

\subsection{$l$-Diversity $k$-Service problem}
The problem is motivated by the $l$-diversity principle, a popular method used for the privacy preservation of public databases. 
%To understand the principle, 
Consider a medical database, in which the data entries are composed of \emph{sensitive} attributes like `name of the disease,' and the \emph{non-sensitive} attributes like `age,' `gender,' `zipcode' etc. The goal is to keep this information anonymous while keeping it meaningful at the same time so that it can be used for research purposes. 
To tackle this problem, Sweeney~\cite{rgather:k_anonymity_2002_Sweeney} proposed the $k$-\emph{anonymity} principle. The principle is very popular, and various privacy-preserving algorithms are based on it~\cite{rgather:k_anonymity_2002_Sweeney,rgather:k_anonymity_2005_Bayardo_Aggarwal,rgather:k_anonymity_2004_LeFevre,rgather:k_anonymity_2005_Aggarwal,rgather:k_anonymity_2004_Meyerson}. Aggarwal \emph{et al.}~\cite{rgather:k_center_2010_Aggarwal} proposed a clustering-based anonymity principle, based on ideas similar to the $k$-\emph{anonymity} principle. 
In this model, the data is grouped into clusters such that each cluster contains at least $r$ entries. Instead of publishing the original data, the cluster centers, cluster sizes, and cluster radius are published to the public. 
In this way, privacy is preserved while keeping data meaningful at the same time. 
This method is popularly known as \emph{r-gather} clustering. 
We have already discussed algorithms for $r$-gather clustering in Subsection~\ref{subsection:B1}. 
The $k$-\emph{anonymity} principle is susceptible to \emph{linking} attacks~\cite{L_diversity:linking_attacks,L_diversity:algorithm_2006_Machanavajjhala} when a cluster contains many similar data-entries. 
To counter this, Machanavajjhala~\emph{et al.}~\cite{L_diversity:algorithm_2006_Machanavajjhala}  introduced the $l$-\emph{diversity} principle. 
According to this principle, a group of data-entries must not contain more than $1/l$ fraction of data entries with the same sensitive attribute. 
We can represent the data items with same sensitive attribute with the same color.
This gives rise to the $l$-\emph{diversity} clustering problem.  
There are many approximation and heuristic algorithms for problems based on the $l$-diversity principle~\cite{L_diversity:algorithm_2006_Machanavajjhala,L_diversity:algorithm_2010_Xiao,L_diversity:algorithm_2007_Ghinita,L_diversity:algorithm_2011_Zhou,L_diversity:algorithm_2013_Dondi}. 
However, these algorithms are not based on the $l$-diversity clustering formulation. 
Unlike the $r$-gather clustering problem, the $l$-diversity clustering problem has not been studied well. 
Li~\emph{et al.}~\cite{L_diversity:2010_kcenter_Li_Jian} gave the first constant approximation algorithm for the $l$-diversity clustering problem corresponding to the $k$-center objective (without any constraint on the number of centers). 
However, they considered a stronger version of the $l$-diversity problem where \emph{all} points must be distinct in a cluster. 
Ding and Xu~\cite{Ding_and_Xu_15}; and Bhattacharya~\emph{et al.}~\cite{bjk18} studied the problem in the Euclidean space and designed a PTAS for fixed value of $k$. 
No constant-approximation is known for the $l$-diversity $k$-means/$k$-median problem in general metric spaces. 
In this paper, we study the problem in metric spaces and design an FPT algorithm that gives $3$-approximation corresponding for the $k$-median objective and $9$-approximation corresponding for the $k$-means objective.

\subsection{Chromatic $k$-Service Problem}
 The problem was formulated by Ding and Xu~\cite{chromatic:kmedian_2012_ding_and_xu} and it has certain applications in cell biology~\cite{chromatic:k_cones_ICALP_2011_ding_and_xu}.
 Ding and Xu~\cite{Ding_and_Xu_15,chromatic:kmedian_2012_ding_and_xu} gave a PTAS for the chromatic $k$-median and $k$-means problems in the Euclidean space (i.e., $C \subseteq L = \mathbb{R}^{d}$) and Bhattacharya {\em et al.}~\cite{bjk18} improved the running time of the algorithm. In this work, we study the problem in general metric spaces. 
 We give an FPT time $(3+\veps)$-approximation algorithm for the chromatic $k$-median problem and $(9+\veps)$-approximation algorithm for the chromatic $k$-means problem.

\dishant{Dishant: Besides this, I couldn't find any concrete application of the Chromatic Clustering problem.
~\\
Also, there is another problem called the ``bi-chromatic $2$-center" problem~\cite{chromatic:kcenter_2015_bicromatic_Arkin}, which was mentioned in the Ding and Xu paper~\cite{Ding_and_Xu_15}. However, this problem is completely different from the chromatic clustering problem, and the authors of~\cite{chromatic:kcenter_2015_bicromatic_Arkin} vaguely argued that it is related to the Chromatic Clustering problem.}

\subsection{Fault-Tolerant $k$-Service Problem}
In certain settings, some facilities (say \emph{servers} in case of a distributed network) may fail after some time. 
In that case, a client must be assigned to some other facility location/server. 
To provide backup against failures, we study the problem in fault-tolerance setting where we consider the cost of a client with respect to multiple facility locations. 
The formal definition of the problem is stated in Table~\ref{table:1}.
    
Jain and Vazirani~\cite{fault:facility_location_2000_jain_Vazirani} gave the first approximation algorithm for the fault-tolerant \emph{facility location} problem. Recall that in facility location the number of open facilities is not bounded but there is a cost of opening a facility.
The algorithm had a $O(\log (l_{m}))$-approximation guarantee, where $l_{m} \coloneqq \max_{x \in C} \{ l_{x} \}$ is the maximum requirement of a client. 
Subsequently, better approximation algorithms were developed that gave constant-approximation guarantees for the problem~\cite{fault:facility_location_2003_guha,fault:facility_location_2003_uniform_Jain,fault:kmedian_2008_uniform_chaitanya,fault:facility_location_2010_byrka}. 
Byrka~\emph{et al.}~\cite{fault:facility_location_2010_byrka} gave a $1.7245$-approximation guarantee that is currently the best-known approximation guarantee for the problem. The fault-tolerant problem is also studied in the \emph{uniform} setting where each client has the same requirement $l$, i.e., $l_{x} = l$ for every $x \in C_{i}$. Swamy and Shmoys~\cite{fault:kmedian_2008_uniform_chaitanya} gave a $1.52$ approximation guarantee for the \emph{uniform} fault-tolerant facility location problem.

For the fault-tolerant \emph{$k$-median} problem, the first approximation algorithm was given by Anthony~\emph{et al.}~\cite{fault:kmedian_2008_non_uniform_anthony}. The algorithm had a $O(\log n)$-approximation guarantee.
Recently, Hajiaghayi~\emph{et al.}~\cite{fault:kmedian_2014_non_uniform_haji_li_SODA} gave an improved $93$-approximation algorithm for the problem.
For the \emph{uniform} fault-tolerant $k$-median problem, a better approximation guarantee of $4$ is known due to Swamy and Shmoys~\cite{fault:kmedian_2008_uniform_chaitanya}. 
In this work, we give an algorithm with an improved approximation guarantee of $3$ for the general case, by allowing an FPT running time. 
This is a significant improvement over the $93$-approximation of  Hajiaghayi~\emph{at.al.}~\cite{fault:kmedian_2014_non_uniform_haji_li_SODA}.
    
For the fault-tolerant \emph{$k$-means} problem, no approximation algorithm is known yet. 
We give the first approximation algorithm for the problem, which gives a $(9+\veps)$-approximation guarantee in FPT time.
    
For the fault-tolerant $k$-center problem, Chaudhuri~\emph{et al.}~\cite{fault:kcenter_1998_naveen} and Khuller~\emph{et al.} \cite{fault:kcenter_2000_khuller} gave a $3$-approximation algorithm. 
However, they considered the objective function where the cost of a client is taken as its distance to the $l_{x}^{th}$ closest facility location, whereas in this work, we consider the cost of a client as the sum of its distances to the $l_{x}$ closest facility locations.

\subsection{Semi-Supervised $k$-Service Problem}
The clustering is typically studied in an unsupervised setting, where the class labels are not known for any data-point. However, in various practical scenarios, we have some background knowledge about the class labels or some kind of supervision in terms of \emph{must-link} and \emph{cannot-link} constraints \cite{semi:2000_ICML_Wagstaff,semi:2001_ICML_Wagstaff,semi:2004_KDD_basu}. Clustering based on this prior knowledge is known as the \emph{semi-supervised} clustering. There are several ways in which we can use the extra information for better clustering of data (see Section 2 of~\cite{semi:Survey_2_Grira}). One such way is to add an extra penalty term to the cost function, for violating the known constraints~\cite{semi:Cost_function_1999_genetic_Demiriz,semi:Cost_function_2006_evolutionary_chakrabarti,semi:cost_function_2007_chi,semi:cost_function_2006_gao}. We define such a cost function for the semi-supervised clustering problem.
Suppose we are given the target clustering $\mathcal{C^{T}} = \{C_{1}', C_{2}', \dotsc, C_{k}'\}$. The goal of the semi-supervised clustering is to output a clustering $\C$ that minimizes the following objective function:
\[
\overline{\Psi}(F,\C) \coloneqq \alpha \cdot \Psi(F,\C) + (1-\alpha) \cdot Dist(\mathcal{C^{T}}, \C), \quad \textrm{for some constant $\alpha \geq 0$}
\]
Here, $Dist(\mathcal{C^{T}}, \C)$ denotes the {\em set-difference} distance.
A similar cost function is considered in the context of \emph{evolutionary clustering}~\cite{semi:Cost_function_2006_evolutionary_chakrabarti}. 
In the evolutionary clustering problem, the task is to output a sequence of clustering for a timestamped data. 
At every time stamp $t$, we have to output a clustering that does not deviate largely from the clustering at the previous time step $(t-1)$. 
In our objective function, the clustering at time step $(t-1)$ acts as the target clustering $\mathcal{C^{T}}$ for the clustering at time step $t$.

Ding and Xu~\cite{Ding_and_Xu_15} gave an FPT algorithm that gives a $(1+\veps)$-approximation for the semi-supervised $k$-median/$k$-means problem in Euclidean space. 
Further, Bhattacharya~\emph{et al.}~\cite{bjk18} improved the running time of the algorithm. 
No approximation algorithm is known for the problem in general metric spaces. 
In this work, we give the first constant-approximation for the problem in general metric spaces. 
Particularly, we give a $(3+\veps)$-approximation and $(9+\veps)$-approximation algorithm for the semi-supervised $k$-median and $k$-means problem, respectively, in FPT time.

\dishant{Dishant: I am not very convinced by the objective function defined by Ding and Xu~\cite{Ding_and_Xu_15}. First, they are assuming that the complete target clustering $\mathcal{C^{T}}$ is given to us. However, we never know the complete class-labels in the standard semi-supervised setting (we only have a partial background knowledge). However, evolutionary clustering is one such type of clustering problem, where complete class labels are known. 
If we consider the evolutionary clustering paper~\cite{semi:Cost_function_2006_evolutionary_chakrabarti}, the $Dist(\mathcal{C^{T}},\C)$ does not signify the set-difference distance. Instead, it is the sum of the Euclidean distances between the centroids of the two clusters. 
Ding and Xu~\cite{Ding_and_Xu_15}, made a reference to the BBG paper for the definition of $Dist(\mathcal{C^{T}},\C)$, while giving the formal definition of the semi-supervised clustering problem. Rather, they should have given a reference to some semi-supervised clustering paper.}

\subsection{Uncertain $k$-service Problem (Probabilistic Clustering)}
Data from  sources such as sensor network data, forecasting data, and demographic data has many uncertainties due to errors and noise in the measured values~\cite{uncertain:survey_2009_Charu}. 
Due to this imprecision, the data is modeled probabilistically rather than deterministically. 
Such imprecise data needs to be clustered for various data-mining purposes and various heuristic algorithms are known for the clustering of uncertain datasets~\cite{uncertain:Heuristic_2006_Chau,uncertain:Heuristic_2006_Chui,uncertain:Heuristic_2005_Kriegal,uncertain:Heuristic_2008_Charu}. 
The problem has also been studied from a theoretical standpoint.
The first theoretical study was done by Cormode and McGregor~\cite{uncertain:kcenter_2008_Cormode}. 
The authors designed approximation algorithms for the \emph{uncertain} $k$-center, $k$-means, and $k$-median problem. 
They defined the input instance in the following manner:
A client $j$ in $C$ is represented by a random variable $X_{j}$ such that $j$ is present at the location $x \in \X$ with probability $t^{j}_{x}$, i.e., $P[X_{j} = x] = t^{j}_{x}$. 
Certainly, we have $\sum_{x \in \X} t^{j}_{x} \leq 1$ for every client $j \in C$.
Also, note that the probability could be less than one since it is possible that a client might not exist at all.
The cost function in defined in two ways: \emph{unassigned} and \emph{assigned}, as follows.
\begin{enumerate}
    \item \textit{Unassigned Cost:} In this case, we output a center-set $F$ that minimizes the following cost function:
    \begin{align*}
    \textrm{$k$-median:} \quad & \sum_{(x_{1},x_{2},\dotsc,x_{n}) \in \X^{n}} \left( \prod_{j = 1}^{n}\textbf{Pr}[X_{j} = x_{j}] \cdot \sum_{j = 1}^{n} d(x_{j},F) \right)\\
    \textrm{$k$-means:} \quad & \sum_{(x_{1},x_{2},\dotsc,x_{n}) \in \X^{n}} \left( \prod_{j = 1}^{n}\textbf{Pr}[X_{j} = x_{j}] \cdot \sum_{j = 1}^{n} d^{2}(x_{j},F) \right)\\
    \textrm{$k$-center:} \quad & \sum_{(x_{1},x_{2},\dotsc,x_{n}) \in \X^{n}} \left( \prod_{j = 1}^{n}\textbf{Pr}[X_{j} = x_{j}] \cdot \max_{j} \left\{ d(x_{j},F)  \right\} \right)
    \end{align*}
    Here, $n \coloneqq |C|$, and $d(x,F) \coloneqq \min_{f \in F} \left\{ d(x,f) \right\}$ denote the distance of $x$ to the closest facility location.
    
    Let $F^{*}$ be an optimal center-set corresponding to the above objective function. We assign a client $j$ to a facility location in $F$, based on its realized position. Suppose $x_{j}$ be the realised position of the client $j$. Then we assign $j$ to a facility location that is closest to $x_{j}$.

    \item \textit{Assigned Cost:} In this case, we assign a client to a cluster center prior to its realization. Therefore, we assume that for a client $j$, all its realizations are assigned to the same center. The goal is to output a center set $F$ and an assignment $\sigma: C \to F$ that minimizes the following cost function:
    \begin{align*}
    \textrm{$k$-median:} \quad & \sum_{(x_{1},x_{2},\dotsc,x_{n}) \in \X^{n}} \left( \prod_{j = 1}^{n}\textbf{Pr}[X_{j} = x_{j}] \cdot \sum_{j = 1}^{n} d(x_{j},\sigma(j)) \right)\\
    \textrm{$k$-means:} \quad & \sum_{(x_{1},x_{2},\dotsc,x_{n}) \in \X^{n}} \left( \prod_{j = 1}^{n}\textbf{Pr}[X_{j} = x_{j}] \cdot \sum_{j = 1}^{n} d^{2}(x_{j},\sigma(j)) \right)\\
    \textrm{$k$-center:} \quad & \sum_{(x_{1},x_{2},\dotsc,x_{n}) \in \X^{n}} \left( \prod_{j = 1}^{n}\textbf{Pr}[X_{j} = x_{j}] \cdot \max_{j} \left\{ d(x_{j},\sigma(j))  \right\} \right)
    \end{align*}
\end{enumerate}
Both of these clustering criteria are useful~\cite{uncertain:kcenter_2008_Cormode,uncertain:kcenter_2018_Alipour}.
For the probabilistic metric $k$-center problem, Cormode and McGregor~\cite{uncertain:kcenter_2008_Cormode} gave bi-criteria approximation algorithms, corresponding to the unassigned objective function. Guha and Munagala~\cite{uncertain:kcenter_2009_Guha} gave the first constant approximation algorithm for the probabilistic metric $k$-center problem (for both the assigned and unassigned cases). Recently, Alipour and Jafari ~\cite{uncertain:kcenter_2018_Alipour} improved the approximation guarantee to 10, for the (assigned) probabilistic $k$-center problem in the general metric spaces. For the Euclidean space (where $C \subseteq L = \mathbb{R}^{d}$) the authors gave an FPT algorithm that has an approximation guarantee of $3+\veps$, for any $\veps>0$.

The above results were for the probabilistic $k$-center problem. Let us discuss the probabilistic $k$-median and $k$-means problem. 
The \emph{unassigned} case of these problems is quite simple. In this case, both problems can be simply reduced to their weighted unconstrained counterparts by linearity of expectation (see Section 5 of ~\cite{uncertain:kcenter_2008_Cormode}). Thus, in the Euclidean space, both problems have $(1+\veps)$-approximation algorithm ~\cite{Feldman07_corset,Chen09_coreset,kumar02,jks14}. In metric space, we get a $(2.675+\veps)$-approximation for the probabilistic $k$-median problem~\cite{byrka15} and $(9+\veps)$-approximation for the probabilistic $k$-means problem~\cite{Svensson17}. Therefore, in this work, we study these problems with respect to their \emph{assigned} objectives only.

First, let us discuss the \emph{assigned} case for the probabilistic $k$-means/$k$-median problems, in the Euclidean space (where $C \subseteq L = \mathbb{R}^{d}$). Cormode and McGregor~\cite{uncertain:kcenter_2008_Cormode} gave an FPT $(1+\veps)$ and $(3+\veps)$-approximation algorithm corresponding to the assigned $k$-means and $k$-median objectives respectively. 
Recently, Ding and Xu~\cite{Ding_and_Xu_15} improved the approximation guarantee for the \emph{assigned} $k$-median problem to $(1+\veps)$, and Bhattacharya~\emph{et al.}~\cite{bjk18} further improved the running time of the algorithm.

Now, let us discuss the (assigned) probabilistic $k$-means/$k$-median problem in the general \emph{metric} spaces. 
Lammersen and Schmidt~\cite{uncertain:streaming_2012_Melanie} gave the first coreset construction for the assigned version of the probabilistic $k$-median problem. Cormode and McGregor~\cite{uncertain:kcenter_2008_Cormode} reduced these problems to their weighted unconstrained counterparts~\cite{uncertain:kcenter_2008_Cormode}, with a certain loss in the approximation factor. In particular, they gave a $(2\alpha+1)$-approximation algorithm for the probabilistic  $k$-median problem and a $(8\alpha+2)$-approximation algorithm for the probabilistic  $k$-means problem \footnote{Cormode and McGregor~\cite{uncertain:kcenter_2008_Cormode} did not state the $(8\alpha+2)$-approximation for the probabilistic $k$-means problem explicitly. However, this result can be obtained using the same technique used to obtain the $(2\alpha+1)$-approximation for the probabilistic $k$-median problem. 
The $(2\alpha+1)$-approximation algorithm for the probabilistic $k$-median problem is stated in Theorem 10 of ~\cite{uncertain:kcenter_2008_Cormode}. In this theorem, if we replace the triangle-inequality with \emph{approximate} triangle-inequality for the $k$-means objective, we would obtain an $(8\alpha+2)$-approximation guarantee.}. Here, $\alpha$ is the approximation guarantee of any unconstrained $k$-median/$k$-means algorithm. The current best approximation guarantee for the unconstrained $k$-median problem is $(2.675+\veps)$~\cite{byrka15}, and the unconstrained $k$-means problem is $(9+\veps)$~\cite{Svensson17}. Substituting these $\alpha$ values, we obtain $(6.35+\veps)$-approximation for the probabilistic $k$-median problem, and $(74+\veps)$-approximation for the probabilistic $k$-means problem. All the above-stated approximation guarantees assume $C \subseteq L = \X$ for the discrete metric spaces. In this work, we improve these approximation guarantees by taking the advantage of FPT running time. We give a $(2+\veps)$ and $(4+\veps)$ for the probabilistic $k$-median and $k$-means problem ($C \subseteq L$), respectively, in FPT time. Moreover, for the general case, where $C$ and $L$ are arbitrary sets, our algorithm gives $(3+\veps)$-approximation and $(9+\veps)$-approximation for the probabilistic $k$-median and $k$-means problem respectively, in FPT time.

\subsection{Outlier $k$-service Problem}
The oultier problem is central to the clustering domain since removing a few outliers (noisy data points) from the data may greatly improves the cost and quality of the clustering.
The notion of the outlier problem was first introduced by Charikar~\emph{et al.}~\cite{outlier:kcenter_2001_Charikar}. 
The authors gave a $3$-approximation algorithm for both the outlier facility location problem and the outlier $k$-center problem.
For the outlier $k$-median problem they gave a bi-criteria approximation algorithm that gives a $4(1+1/\veps)$-approximation guarantee while violating the number of outliers by a factor of $(1+\veps)$. Chen~\cite{outlier:kmedian_2008_Chen}
gave the first constant-approximation algorithm for the outlier $k$-\emph{median} problem. 
Recently, Ravishankar~\emph{et al.}~\cite{outlier:kmeans_2018_Ravishankar} improved the approximation guarantee for the outlier $k$-median to $(7.081+\veps)$. Moreover, they gave a $53.002$-approximation algorithm for the oulier $k$-\emph{means} problem. Friggstad~\emph{et al.}~\cite{outlier:kmeans_2018_Friggstad_bicriteria} gave a bi-criteria algorithm for the outlier $k$-means problem that gives a $(25+\veps)$-approximation guarantee and opens $(1+\veps)$ facilities. 
 
All the above-stated algorithms are polynomial time algorithms.
Now, let us discuss FPT algorithms for the problem. Feng~\emph{et al.}~\cite{outlier:FPT_2019_ISAAC_Feng} gave a $(6+\veps)$-approximation for the outlier $k$-means problem with $C \subseteq L$, with an FPT time of $O \left( n \cdot \beta^{k} \left( \frac{k+m}{\veps} \right) ^{k} \right) $, for some constant $\beta>0$. 
In this work, we improve on this result by giving a $(4+\veps)$-approximation algorithm for the oulier $k$-means problem with $C \subseteq L$ and a $(2+\veps)$-approximation algorithm for the outlier $k$-median problem with $C \subseteq L = \X$ with a running time of $O \left( n \cdot \left( \frac{k+m}{\veps} \right) ^{O(k)} \right)$. 
For the general case where $C$ and $L$ are arbitrary sets, our algorithm gives a $(9+\veps)$-approximation guarantee for the outlier $k$-means problem and a $(3+\veps)$-approximation guarantee for the outlier $k$-median problem with the same running time.
 
 For the Euclidean version of the problem (i.e., $C \subseteq L = \mathbb{R}^{d}$), Feldman and Schulmany~\cite{outlier:FPT_2012_SODA_Feldman} gave a PTAS for the outlier $k$-median problem with an FPT time of \\
 $O\left( nd \cdot (m+k)^{O(m+k)} + (\veps^{-1}  k \log n)^{O(1)} \right)$. 
 Recently, Feng~\emph{et al.}~\cite{outlier:FPT_2019_ISAAC_Feng} improved this running time to $O\left( nd \cdot \left( \frac{m+k}{\veps} \right) ^{\left(k/\veps \right)^{O(1)}} \right)$.

\section{Preliminaries}
We give a few notations and identities that we will use often in our discussions. 
We define the unconstrained $k$-service cost of a set $S$ with respect to a center set $F$ as:
$$\Phi(F,S) \coloneqq \sum_{x \in S} \min_{f \in F} d^{\ell}(f,x).$$ 
For a singleton set $\{f\}$, we denote $\Phi(\{f\},S)$ by $\Phi(f,S)$. 
We denote the optimal (unconstrained) $k$-service cost of a instance by $OPT(L,C)$.
Let us now look at some of the identities that we will use in our analysis. 
Following is the binomial approximation technique that we will use to simplify terms with large exponents.
\begin{fact}[Binomial Approximation] For $\veps \cdot n \leq 1/2$, we have $(1+\veps)^{n} \leq (1+2\veps n)$
\label{fact:bin_approx}
\end{fact}

\noindent We use the next fact to carry out the trade-off between two values $a$ and $b$. 
We will choose the value of $\delta$ according to our requirement.
\begin{fact}\label{fact:trade_off}
For any $\delta > 0$, we have $(a+b)^{\ell} \leq (1+\delta)^{\ell} \cdot b^{\ell} + \left(1+\frac{1}{\delta}\right) ^{\ell} \cdot a^{\ell}$.
\end{fact}
\begin{proof}
There are two possibilities: $a \leq \delta \cdot b$ or $a > \delta \cdot b$.
For the first case, we have $(a+b)^{\ell} \leq (1+\delta)^{\ell} \cdot b^{\ell}$. For the second case, we have $(a+b)^{\ell} \leq \left( 1+\frac{1}{\delta} \right) ^{\ell} \cdot a^{\ell}$. Hence we get the required result. 
\end{proof}

\noindent Since we are working with metric spaces, triangle inequality becomes a powerful tool for analysis. 
We need to generalize the triangle inequality since we are dealing with a general cost function $d^{\ell}$. 
The following inequality is the generalization of the triangle inequality and simply follows from the \emph{power-mean} inequality.
\begin{fact}[Approximate triangle inequality] For a set of points $\{a,b,c\} \in \X$, $d^{\ell}(a,b) \leq 2^{\ell-1} \cdot\left( (d^{\ell}(a,c) + d^{\ell}(c,b) \right)$. Similarly for a set of four points $\{ a,b,c,d \} \in \X$ we have $d^{\ell}(a,b) \leq 3^{\ell-1} \cdot\left( (d^{\ell}(a,c) + d^{\ell}(c,d) + d^{\ell}(d,b) \right)$
\label{fact:approx_tri_ineq}
\end{fact}

\noindent Let $OPT(C,C)$ denote the optimal cost of a \emph{unconstrained} $k$-service instance when facilities are only allowed to open at client locations. 
The following fact easily follows from the \emph{power-mean} inequality (for the detailed proof see Theorem 2.1 of~\cite{streaming_2000_FOCS}).
\begin{fact}\label{fact:optimal}
$OPT(C,C) \leq 2^{\ell} \cdot OPT(L,C)$
\end{fact}

%\noindent Following is the main lemma of this paper that shows the power of \emph{uniform sampling} in obtaining a constant-approximation for a cluster $S$. 
\noindent The following lemma demonstrates the effectiveness of uniform sampling in obtaining constant factor approximation.
This lemma (or a similar version) has been used in multiple other works in analysing sampling based algorithms.
\begin{lemma}~\label{lemma:exp}
Let $S \subseteq C$ be any subset of clients and let $f^{*}$ be any center in $L$. 
If we uniformly sample a point $x$ in $S$ and open a facility at the closest location in $L$, then the following identity holds:
\[
\mathbb{E}[\Phi(t(x),S)] \leq 3^{\ell} \cdot \Phi(f^{*},S),
\]
where $t(x)$ is the closest facility location from $x$.
\end{lemma}
\begin{proof}
The proof follows from the following sequence of inequalities.
\begin{align*}
    \mathbb{E}[\Phi(t(x),S)] &= \frac{1}{|S|} \left( \sum_{x \in S} \Phi(t(x),S) \right)\\
    &= \frac{1}{|S|} \left( \sum_{x \in S} \sum_{x' \in S} d^{\ell}(t(x),x') \right) \\
    &\leq \frac{3^{\ell-1}}{|S|} \left( \sum_{x \in S}  \sum_{x' \in S}  \left(d^{\ell}(f^{*},x') + d^{\ell}(x,f^{*}) + d^{\ell}(t(x),x) \right) \right), \quad (\textrm{Using Fact~\ref{fact:approx_tri_ineq}}) \\
    &\leq \frac{3^{\ell-1}}{|S|} \left( \sum_{x \in S}  \sum_{x' \in S}  \left(d^{\ell}(f^{*},x') + d^{\ell}(x,f^{*}) + d^{\ell}(f^{*},x) \right) \right), \hspace{6.5mm} (\textrm{Using defn. of $t(x)$})\\
    &= \frac{3^{\ell-1}}{|S|} \left( \sum_{x \in S} \Phi(f^{*},S) +  \sum_{x' \in S} \Phi(f^{*},S) + \sum_{x' \in S} \Phi(f^{*},S) \right) \\
    &= \frac{3^{\ell-1}}{|S|} \bigg( 3 |S| \cdot \Phi(f^{*},S) \bigg) \\
    &= 3^{\ell} \cdot \Phi(f^{*},S)
\end{align*}
This completes the proof of the lemma.
\end{proof}
\noindent 
%Note that in the above lemma, it is not necessary to open a facility at the closest location from $x$. Rather, we can open a facility at a location, which is at least as close to $x$ as $f^{*}$. 
In the next lemma, we show that the flexibility to open a facility at any client location gives a better approximation guarantee.
\begin{lemma}~\label{lemma:exp2}
Let $S \subseteq C$ be any subset of clients and let $f^{*}$ be be any center in $L$. If we uniformly sample a point $x$ in $S$ and open a facility at $x$, then the following identity holds:
\[
\mathbb{E}[\Phi(x,S)] \leq 2^{\ell} \cdot \Phi(f^{*},S).
\]
\end{lemma}
\begin{proof} The proof follows from the following inequalities.
\begin{align*}
    \mathbb{E}[\Phi(x,S)] &= \frac{1}{|S|} \left( \sum_{x \in S} \Phi(x,S) \right)\\
    &= \frac{1}{|S|} \left( \sum_{x \in S} \sum_{x' \in S} d^{\ell}(x,x') \right) \\
    &\leq \frac{2^{\ell-1}}{|S|} \left( \sum_{x \in S}  \sum_{x' \in S}  \left(d^{\ell}(f^{*},x') + d^{\ell}(x,f^{*}) \right) \right), \quad (\textrm{Using Fact~\ref{fact:approx_tri_ineq}}) \\
    &= \frac{2^{\ell-1}}{|S|} \bigg( 2 |S| \cdot \Phi(f^{*},S) \bigg) \\
    &= 2^{\ell} \cdot \Phi(f^{*},S)
\end{align*}
This completes the proof of the lemma.
\end{proof}

\section{Algorithm for List \texorpdfstring{$k$}{}-Service}
\label{section:list_k_service}
In this section, we design and analyze the list $k$-service algorithm. 
As described earlier, an FPT algorithm for the list $k$-service problem gives an FPT algorithm for a constrained version of the $k$-service problem that has an efficient or FPT-time partition algorithm. 
Our sampling based algorithm is similar to the algorithm of Goyal \emph{et al.}~\cite{gjk19} that was specifically designed for the Euclidean setting.
However, the analysis differs at various steps since we study the problem in general metric spaces, whereas Goyal \emph{et al.}~\cite{gjk19} studied the problem in the Euclidean space where $C \subseteq L = \mathbb{R}^{d}$. 
Following is our algorithm for the list $k$-service problem:

\begin{Algorithm}[h]
\begin{framed}
\LK ~($C, L, k, d, \ell, \veps$) \vspace{1mm}\\
\hspace*{0.3in} {\bf Inputs}: $k$-service instance $(C,L,k,d,\ell)$ and accuracy $\veps$ \\
\hspace*{0.3in} {\bf Output}: A list $\mathcal{L}$, each element in $\mathcal{L}$ being a $k$-center set \vspace*{2mm}\\
\hspace*{0.3in} {\bf Constants}: $\beta = 4^{\ell-1} \cdot \left( \myfrac[2pt]{\ell^{\ell} \cdot 3^{\ell^{2}+4\ell+3}}{\veps^{ \, \ell+1}}+1 \right)$; $\gamma = \myfrac[2pt]{
\ell^{\ell} \cdot 3^{\ell^{2} + 5\ell+1}}{\veps^{\, \ell}}$; $\mathbf{\eta} = \myfrac[2pt]{\alpha \, \beta \, \gamma \, k \cdot 3^{\ell+2}}{\veps^{2}}$\\[4pt]
\hspace*{0.1in} (1) \ \ \ Run any $\alpha$-approximation algorithm for the \emph{unconstrained} $k$-service  \\
\hspace*{0.1in} \ \ \ \ \ \ \ \ instance $(C,C,k,d,\ell)$ and let $F$ be the obtained center-set. \\
\hspace*{0.6in} ({\it $k$-means++ ~\cite{kmeanspp} is one such algorithm.})\\
\hspace*{0.1in} (2) \ \ \ $\mathcal{L} \gets \emptyset$\\
\hspace*{0.1in} (3) \ \ \ Repeat $2^k$ times:\\
\hspace*{0.1in} (4)\hspace*{0.3in}  \ \ \ Sample a multi-set $M$ of $\eta k$ points from $C$ using $D^{\ell}$-sampling w.r.t. \\
\hspace*{0.1in} \hspace*{0.5in} \ \ \ center set $F$\\
\hspace*{0.1in} (5)\hspace*{0.3in}  \ \ \ $M \gets M \cup F$ \\
\hspace*{0.1in} (6)\hspace*{0.3in}  \ \ \ $T \gets \emptyset$ \\
\hspace*{0.1in} (7)\hspace*{0.3in} \ \ \ For every point $x$ in $M$:\\
\hspace*{0.1in} (8)\hspace*{0.9in} 
$T \gets T \cup \{k \text{ points in $L$ that are closest to $x$}\}$\\
\hspace*{0.1in} (9)\hspace*{0.3in} \ \ \ For all subsets $S$ of $T$ of size $k$:\\
\hspace*{0.1in} (10)\hspace*{0.84in} $\mathcal{L} \gets \mathcal{L} \cup \{ S\}$\\
\hspace*{0.1in} (11) \ \  return($\mathcal{L}$)
\end{framed}
\vspace*{-4mm}
\caption{Algorithm for the list $k$-service problem}
\label{algorithm:kmeans}
\end{Algorithm}

\noindent We start with the main intuition in the next subsection before going into the details of the proof.

\subsection{Algorithm Description and Intuition}
In the first step, we obtain a center-set $F \subseteq C$, which is an $\alpha$-approximation for the \emph{unconstrained} $k$-service instance $(C,C,k,d,\ell)$.
That is,
$$\Phi(F,C) \leq \alpha \cdot OPT(C,C).$$ 
One such algorithm is the $k$-means++ algorithm~\cite{kmeanspp} that gives an $O(4^{\ell} \cdot \log k)$-approximation guarantee and a running time $O(nk)$.
Now, let us see how set $F$ can help us. 
Let us focus on any cluster $C_{i}$ of a target clustering $\C = \{C_{1}, \dotsc, C_{k}\}$. 
Our main objective would be to uniformly sample a point from $C_{i}$, so that we could achieve a constant approximation for $C_{i}$  using Lemma~\ref{lemma:exp}. 
We will do a case analysis based on the distance of points in $C_i$ from the nearest point in $F$.
%There are two possibilities for the points in $C_{i}$. 
Consider the following two possibilities:
The first possibility is that the points in $C_{i}$ are close to $F$. 
If this is the case, we can uniformly sample a point from $F$ instead of $C_{i}$. 
Note that we cannot uniformly sample from $C_i$ even if we wanted to since $C_{i}$ is not known to us. 
This would incur some extra cost. 
However, the cost is small and can be bounded easily. 
To cover this first possibility, the algorithm adds the entire set $F$ to the set of sampled points $M$ (see line (5) of the algorithm).
The second possibility is that the points in $C_{i}$ are far-away from $F$. 
In this case, we can $D^{\ell}$-sample the points from $C$.
Since the points in $C_{i}$ are far away, the sampled set would contain a good portion of points from $C_{i}$ and the points will be almost uniformly distributed. 
We will show that almost uniform sampling is sufficient to apply lemma~\ref{lemma:exp} on $C_{i}$. 
However, we would have to sample a large number of points to boost the success probability. 
This requirement is taken care of by line (4) of the algorithm.
Note that we may need to use a hybrid approach for analysis since the real case may be a combination of the first and second possibility.

To apply lemma~\ref{lemma:exp}, we need to fulfill one more condition, i.e., we need the closest facility location from a sampled point. This requirement is handled by lines (7) and (8) of the algorithm. However, note that the algorithm picks $k$-closest facility locations instead of just one facility location. We will show that this step is crucial to obtain a \emph{hard-assignment} solution for the problem. At last, the algorithm adds all the potential center sets to a list $\mathcal{L}$ (see line (9) and (10) of the algorithm). The algorithm repeats this procedure $2^{k}$ times to boost the success probability (see line (3) of the algorithm).
We will show the following main result.

\begin{theorem}\label{theorem:list_k_service}
Let $0 < \veps \leq 1$.
Let $(C,L,k,d,\ell)$ be any $k$-service instance and let $\mathcal{C} = \{C_{1}, C_{2},\dotsc,C_{k}\}$ be any arbitrary clustering of the client set.
%For a $k$-service instance $(C,L,k,d,\ell)$ and for an arbitrary clustering $\mathcal{C} = \{C_{1}, C_{2},\dotsc,C_{k}\}$ of the client set, 
The algorithm ~ \LK($C,L,k,d,\ell,\veps$), with probability at least $1/2$ outputs a list $\mathcal{L}$ of size $(k /\veps)^{O(k \, \ell^{\,2})}$, such that there is a $k$ center set $S \in \mathcal{L}$ in the list such that
\[
    \Psi (S,\C) \leq (3^{\ell} + \veps) \cdot \Psi^{*}(\C).
\]
Moreover, the running time of the algorithm is $O \left( n \cdot (k /\veps)^{O(k \, \ell^{\,2})} \right)$.
\end{theorem}

\subsection{Analysis}
Let $\C = \{C_{1}, C_{2}, \dotsc, C_{k}\}$ be the (unknown) target clustering and $F^{*} = \{f_{1}^{*}, f_{2}^{*}, \dotsc, f_{k}^{*}\}$ be the corresponding optimal center set.
Let $\Delta(C_{i})$ denote the cost of a cluster $C_{i}$ with respect to $F$, i.e., $\Delta(C_{i}) = \Phi(f_{i}^{*},C_{i})$.  Let us classify the clusters into two categories: $W$ and $H$.
\[
W \coloneqq \{C_{i} \mid \Phi(F,C_{i}) \leq \frac{\veps}{\alpha \, \gamma \, k} \cdot \Phi(F,C
), \textrm{ for } 1 \leq i \leq k\}
\]
\[
H \coloneqq \{C_{i} \mid \Phi(F,C_{i}) > \frac{\veps}{ \alpha \, \gamma \, k} \cdot \Phi(F,C), \textrm{ for } 1 \leq i \leq k\}
\]
In other words, $W$ contains the \emph{low-cost} clusters and $H$ contains the \emph{high-cost} clusters with respect to $F$. Now, let us look at the set $M$ obtained by lines (4) and (5) of the algorithm. The set $M$ contains some $D^{\ell}$-sampled points from $C$ and the center set $F$. We show that $M$ has the following property.
\begin{quote}{\textbf{Property}-\boldmath{$\rm I$}:}\label{property:1}
For any cluster $C_{i} \in \{C_{1}, C_{2}, \dotsc, C_{k}\}$,
with probability at least $1/2$, there is a point $s_{i}$ in $M$ such that such that the following holds:
\[
    \Phi(t(s_{i}),C_{i}) \leq 
\begin{cases}
    \Big( 3^{\ell} + \fdfrac{\veps}{2} \Big) \cdot \Delta(C_{i}) + \fdfrac{ \veps}{2^{\ell+1} \, k} \cdot OPT(C,C),& \textrm{if } C_{i} \in W \\[7pt]
    \Big( 3^{\ell}+\fdfrac{\veps}{2}
    \Big) \cdot \Delta(C_{i}) ,& \textrm{if } C_{i} \in H
\end{cases}
\]
where $t(s_{i})$ denotes any facility location that is closer to $s_i$ than $f_{i}^{*}$, i.e., $d(s_{i},t(s_{i})) \leq d(s_{i},f_{i}^{*})$.
\end{quote}

\noindent First, let us see how this property gives the desired result. By Fact~\ref{fact:optimal}, we have  $OPT(C,C) \leq 2^{\ell} \cdot OPT(L,C)$. Moreover, the optimal cost $OPT(L,C)$ of the unconstrained $k$-service instance is always less than the constrained $k$-service cost $\sum_{i=1}^{k} \Delta(C_{i})$. Therefore, Property-I implies that $T_{s} \coloneqq \{t(s_{1}), t(s_{2}), \dotsc, \allowbreak t(s_{k})\}$ is a $\left( 3^{\ell} + \veps \right)$-approximation for $\C$, with probability at least $1/2^{k}$.
Now, note that, the facility locations that are closest to $s_{i}$ satisfy the definition of $t(s_{i})$. Moreover, the algorithm adds one such facility location to set $T$ (see line (8) of the algorithm).
Thus there is a center-set $T_{s}$ in the list that gives $(3^{\ell} + \veps)$-approximation for $\C$.  To boost the success probability to $1/2$, the algorithm repeats the procedure $2^{k}$ times (see line (3) of the algorithm). Based on these arguments, it looks like we got the desired result.
However, there is one issue that we need to take care of.
Remember, we are looking for a hard assignment for the problem, and the set $T_{s}$ could be a soft center-set, since the closest facility locations might be same for $s_{i}$'s. 
In other words, $t(s_{i})$ could be same as $t(s_{j})$ for some $i \neq j$. 
At the end of this section we will show that there is indeed a hard center-set in the list $\mathcal{L}$, that gives the required approximation for the problem.
For now let us prove Property-I for $M$ and the target clusters. 
First consider the case of low-cost clusters as follows.\\

\noindent \textbf{Case 1: \boldmath{$\Phi(F,C_{i}) \leq \fdfrac{\veps}{\alpha \, \gamma \, k} \cdot \Phi(F,C)$}}
~\\~\\
For a point $x \in \X$, let $c(x)$ denote the closest location in $F$. 
Based on this definition, consider a multi-set $M_{i} \coloneqq \{ c(x) \mid x \in C_{i} \}$. 
Since $C_{i}$ has a low cost with respect to $F$, the points in $C_{i}$ are close to from points from $F$. 
%Since $C_{i}$ is not known to us 
Consider uniformly sampling a point from $M_{i}$. 
In the next lemma, We show that a uniformly sampled point from $M_i$ is a good enough center for $C_{i}$.
\begin{lemma}~\label{lemma:close_cluster}
Let $p$ be a point sampled uniformly at random from $M_i$. Then the following bound holds:
%If we sample a point $p \in M_{i}$, uniformly at random, then the following identity holds:
\[
    \mathbb{E}[\Phi(t(p),C_{i})] \leq \left( 3^{\ell} + \frac{\veps}{2} \right) \cdot \Delta(C_{i}) + \frac{ \veps}{2^{\ell+1} \, k} \cdot OPT(C,C).
\]
\end{lemma}

\begin{proof}
The proof follows from the following sequence of inequalities.
\begin{align*}
    \mathbb{E}[\Phi(t(p),C_{i})] &= \frac{1}{|C_{i}|} \cdot \left( \sum_{p \in M_{i}} \Phi(t(p),C_{i}) \right) \\
    &= \frac{1}{|C_{i}|} \cdot \left(  \sum_{p \in M_{i}} \sum_{x \in C_{i}} d^{\ell}(x,t(p))  \right) \\
    &= \frac{1}{|C_{i}|} \cdot \left( \sum_{x' \in C_{i}} \sum_{x \in C_{i}} d^{\ell}(x,t(c(x'))) \right) \\
    &\leq \frac{1}{|C_{i}|} \cdot \left( \sum_{x' \in C_{i}} \sum_{x \in C_{i}} \left( d(x,x') + d(x',c(x')) + d(c(x'),t(c(x'))) \right)^{\ell} \right), \\
    &\qquad \qquad \qquad \qquad \qquad \qquad \qquad \qquad \textrm{by triangle inequality} \\
    &\leq \frac{1}{|C_{i}|} \cdot \left( \sum_{x' \in C_{i}} \sum_{x \in C_{i}} \left( d(x,x') + d(x',c(x')) + d(c(x'),f^{*}_{i}) \right)^{\ell} \right), \\
    &\qquad \qquad \qquad \qquad \qquad \qquad \qquad \qquad \textrm{by the defn. of $t(x)$} \\
    &\leq \frac{1}{|C_{i}|} \cdot \left( \sum_{x' \in C_{i}} \sum_{x \in C_{i}} \left( d(x,x') + d(x',c(x')) + d(x',c(x')) + d(x',f_{i}^{*}) \right)^{\ell} \right), \\
    &\qquad \qquad \qquad \qquad \qquad \qquad \qquad \qquad \textrm{by triangle inequality}
\end{align*}
Let us use Fact~\ref{fact:trade_off}, by setting $a = 2 \cdot d(x',c(x'))$ and $b = d(x,x') + d(x',f_{i}^{*})$. We get the following expression:
\begin{align*}
    \mathbb{E}[\Phi(t(p),C_{i})] &\leq \frac{1}{|C_{i}|} \cdot \left( \sum_{x' \in C_{i}} \sum_{x \in C_{i}} \left( \left(1+\frac{1}{\delta} \right)^{\ell} \cdot (d(x,x') + d(x',f_{i}^{*}))^{\ell} + (1+\delta)^{\ell} \cdot 2^{\ell} \cdot d^{\ell}(x',c(x')) \right) \right), \\
    &\hspace{100mm} \textrm{for any $\delta > 0 $} \\
    &=  \left(1+\frac{1}{\delta} \right)^{\ell} \cdot \frac{1}{|C_{i}|} \cdot \left(  \sum_{x' \in C_{i}} \sum_{x \in C_{i}} (d(x,x') + d(x',f_{i}^{*}))^{\ell} \right) + \\
    &\qquad \qquad \qquad \qquad \qquad (1+\delta)^{\ell} \cdot \frac{1}{|C_{i}|} \cdot \left( \sum_{x' \in C_{i}} |C_{i}| \cdot  2^{\ell} \cdot d^{\ell}(x',c(x')) \right)
\end{align*}
By lemma~\ref{lemma:exp}, we have $\mathbb{E}[\Phi(t(x),C_{i})] \leq \frac{1}{|C_{i}|} \cdot \Big( \sum_{x' \in C_{i}} \sum_{x \in C_{i}} (d(x,x') + d(x',f_{i}^{*}) )^{\ell} \Big) \leq 3^{\ell} \cdot \Delta(C_{i})$. Thus, we get:

\begin{align*}
    \mathbb{E}[\Phi(t(p),C_{i})] &\leq \left(1+\frac{1}{\delta} \right)^{\ell} \cdot 3^{\ell} \cdot \Delta(C_{i}) + (1+\delta)^{\ell} \cdot \frac{1}{|C_{i}|} \cdot \left( \sum_{x' \in C_{i}} |C_{i}| \cdot  2^{\ell} \cdot d^{\ell}(x',c(x')) \right) \\
    &= \left(1+\frac{1}{\delta} \right)^{\ell} \cdot 3^{\ell} \cdot \Delta(C_{i}) + 2^{\ell} (1+\delta)^{\ell} \cdot \Phi(F,C_{i}) \\
    &= \left(1+\frac{\veps}{\ell \cdot  3^{\ell+2}} \right)^{\ell} \cdot 3^{\ell} \cdot \Delta(C_{i}) + 2^{\ell}(1+\delta)^{\ell} \cdot \Phi(F,C_{i}), \hspace{9mm} \textrm{by substituting $\delta = \frac{\ell \cdot 3^{\ell+2}}{\veps}$ }\\
    &\leq \left(1+2 \ell \cdot \frac{\veps}{\ell \cdot 3^{\ell+2}} \right) \cdot 3^{\ell} \cdot \Delta(C_{i}) + 2^{\ell} (1+\delta)^{\ell} \cdot \Phi(F,C_{i}), \hspace{4mm} \textrm{by Fact~\ref{fact:bin_approx}} \\
    &\leq \left( 3^{\ell} + \frac{\veps}{2} \right) \cdot \Delta(C_{i}) + 2^{\ell} (1+\delta)^{\ell} \cdot \Phi(F,C_{i}), \\
    &\leq \left( 3^{\ell} + \frac{\veps}{2} \right) \cdot \Delta(C_{i}) + 2^{\ell} (2\delta)^{\ell} \cdot \Phi(F,C_{i}), \hspace{11mm} \textrm{$\because 1 \leq \delta$, for $\veps \leq 1$}\\
    &\leq \left( 3^{\ell} + \frac{\veps}{2} \right) \cdot \Delta(C_{i}) + \frac{4^{\ell} \cdot \delta^{\ell} \cdot \veps}{\alpha \, \gamma \, k} \cdot \Phi(F,C), \hspace{9mm} \textrm{$\because \Phi(F,C_{i}) \leq \frac{\veps}{\alpha \, \gamma \, k} \cdot \Phi(F,C)$}\\
    &\leq \left( 3^{\ell} + \frac{\veps}{2} \right) \cdot \Delta(C_{i}) + \frac{\veps}{2^{\ell+1} \, \alpha \, k} \cdot \Phi(F,C), \hspace{10mm} \textrm{$\because \gamma = \frac{\ell^{\ell} \cdot 3^{\ell^{2} + 5\ell + 1}}{\veps^{\ell}}$}\\
    &\leq \left( 3^{\ell} + \frac{\veps}{2} \right) \cdot \Delta(C_{i}) + \frac{\veps}{2^{\ell+1} \, k} \cdot OPT(C,C), \hspace{7mm} \textrm{$\because \Phi(F,C) \leq \alpha \cdot OPT(C,C)$} \\
\end{align*}
This completes the proof of the lemma.
\end{proof}

\noindent Since the above lemma estimates the average cost corresponding to a sampled point, there has to be a point $p$ in $M_{i}$ such that $\Phi(t(p),C_{i}) \leq \left( 3^{\ell} + \fdfrac{\veps}{2} \right) \cdot \Delta(C_{i}) + \fdfrac{\veps}{2^{\ell+1}\,k} \cdot OPT(C,C)$. 
Since $M_{i}$ is only composed of the points from $F$ and we keep the entire set $F$ in $M$ (see line (5) of the algorithm), therefore Property-I is satisfied for every cluster $C_{i} \in W$. 
Let us now prove Property I for the high cost clusters.
~\\ \\
\textbf{Case 2: \boldmath{$\Phi(F,C_{i}) > \fdfrac{\veps}{\alpha \, \gamma \, k} \cdot \Phi(F,C)$}}. ~\\~\\
Since the cost of the cluster is high, some points of $C_{i}$ are far away from the center set $F$. 
%To capture this property, 
We partition $C_{i}$ into two sets: $C^{n}_{i}$ and $C_{i}^{f}$, as follows.
\[
C_{i}^{n} \coloneqq \{ x \mid d^{\ell}(c(x),x) \leq R^{\ell}, \text{ for } x \in C_{i} \}, \quad \textrm{where } R^{\ell} = \frac{1}{\beta} \cdot \frac{\Phi(F,C_{i})}{|C_{i}|}
\]
\[
C_{i}^{f} \coloneqq \{ x \mid d^{\ell}(c(x),x) > R^{\ell}, \text{ for } x \in C_{i} \}, \quad \textrm{where } R^{\ell} = \frac{1}{\beta} \cdot \frac{\Phi(F,C_{i})}{|C_{i}|}
\]
In other words, $C_{i}^{n}$ represents the set of points that are \emph{near} to the center set $F$ and $C_{i}^{f}$ represents the set of points that are \emph{far} from the center set $F$. Recall that our prime objective is to obtain a uniform sample from $C_{i}$, so that we can apply lemma~\ref{lemma:exp}. 
To achieve that we consider sampling from $C_{i}^{n}$ and $C_{i}^{f}$ separately. The idea is as follow. 
To sample a point from $C_{i}^{f}$ we use the $D^{\ell}$-sampling technique and show that it gives an almost uniform sample from $C_{i}^{f}$. 
For $C_{i}^{n}$, we will use $F$ as its proxy, and sample a point from $F$ instead. However, doing so would incur an extra cost. We will show that the extra cost is proportional to $\Phi(F,C_{i}^{n})$, which can be bounded easily. To bound the extra cost we will use the following lemma.

\begin{lemma}\label{lemma:radius}
For $R^{\ell} = \myfrac[1pt]{1}{\beta} \cdot \myfrac[2pt]{\Phi(F,C_{i})}{|C_{i}|}$, we have $\Phi(F,C_{i}^{n}) \leq \myfrac[2pt]{\veps^{\ell+1}}{\ell^{\, \ell} \cdot 3^{\, \ell^{2}+5\ell+2}} \cdot \Delta(C_{i})$.
\end{lemma}
\begin{proof}
We have,
\begin{align*}
\Delta(C_{i}) &\geq \Phi(f_{i}^{*},C^{n}_{i})\\
&= \sum_{x \in C^{n}_{i}} d^{\ell}(f_{i}^{*},x) \\
&\geq \sum_{x \in C^{n}_{i}} \left( \frac{d^{\ell}(c(x),f_{i}^{*})}{2^{\ell-1}} - d^{\ell}(x,c(x))  \right), \quad \textrm{by Fact~\ref{fact:approx_tri_ineq}} \\
&=  \sum_{x \in C^{n}_{i}} \left( \frac{d^{\ell}(c(x),f_{i}^{*})}{2^{\ell-1}} \right) - \Phi(F,C_{i}^{n}) \\
&\geq  \sum_{x \in C^{n}_{i}} \left( \frac{d^{\ell}(c(f_{i}^{*}),f_{i}^{*})}{2^{\ell-1}} \right) - \Phi(F,C_{i}^{n})
\end{align*}
Using Fact~\ref{fact:approx_tri_ineq}, we get $\Phi(c(f_{i}^{*}),C_{i}) \leq 2^{\ell-1} \cdot \left( \Delta(C_{i}) + |C_{i}|\cdot d^{\ell}(c(f_{i}^{*}),f_{i}^{*}) \right)$.
Since $\Phi(F,C_{i}) \leq \Phi(c(f_{i}^{*}),C_{i})$, we get $d^{\ell}(c(f_{i}^{*}),f_{i}^{*}) \geq \fdfrac{\Phi(F,C_{i}) - 2^{\ell-1} \cdot \Delta(C_{i})}{2^{\ell-1}|C_{i}|}$. Using this, the previous expression simplifies to:
\begin{align*}
\Delta(C_{i}) &\geq |C^{n}_{i}| \left( \frac{\Phi(F,C_{i}) - 2^{\ell-1} \cdot \Delta(C_{i})}{4^{\ell-1} \cdot |C_{i}|} \right) - \Phi(F,C_{i}^{n}) \\
&= |C_{i}^{n}| \cdot \frac{\beta R^{\ell}}{4^{\ell-1}}   - |C^{n}_{i}| \cdot \frac{\Delta(C_{i})}{2^{\ell-1} \cdot |C_{i}|} - \Phi(F,C_{i}^{n}), \hspace{15mm} \because R^{\ell} = \frac{1}{\beta} \cdot \frac{\Phi(F,C_{i})}{|C_{i}|}\\
&\geq \Phi(F,C_{i}^{n}) \cdot \frac{\beta}{4^{\ell-1}}   - |C^{n}_{i}| \cdot \frac{\Delta(C_{i})}{2^{\ell-1} \cdot |C_{i}|} - \Phi(F,C_{i}^{n}), \hspace{7mm} \because \Phi(F,C_{i}^{n}) \leq |C_{i}^{n}| \cdot R^{\ell}\\
&\geq \frac{(\beta-4^{\ell-1})}{4^{\ell-1}} \Phi(F,C_{i}^{n}) - \Delta(C_{i}), \hspace{36mm} \because |C_{i}^{n}| \leq |C_{i}| \leq 2^{\ell-1} \cdot |C_{i}|
\end{align*}
On rearranging the terms of the expression, we get 
\begin{align*}
    \hspace{5mm}\Phi(F,C_{i}^{n}) &\leq \frac{2 \cdot 4^{\ell-1} }{\beta - 4^{\ell-1}} \cdot \Delta(C_{i})\\
    &\leq \myfrac[2pt]{\veps^{\ell+1}}{\ell^{\ell} \cdot 3^{\ell^{2}+5\ell+2}} \cdot \Delta(C_{i}) \hspace{10mm} \because \beta = 4^{\ell-1} \cdot \left( \myfrac[2pt]{\ell^{\ell} \cdot 3^{\ell^{2}+4\ell+3}}{\veps^{\ell+1}}+1 \right)
\end{align*}
Hence proved.
\end{proof}

Now, let us prove the main result. 
We need to define a few things.
Since we are using $F$ as a proxy for $C_{i}^{n}$, we define a multi-set $M_{i}^{n} \coloneqq \{ c(x) \mid x \in C_{i}^{n} \}$. Let us define another multi-set $M_{i} \coloneqq C_{i}^{f} \cup M_{i}^{n}$. 
In the following lemma we show that there is a point in $M_{i}$ that is a good center for $C_{i}$. The lemma is similar to lemma~\ref{lemma:close_cluster} of the low-cost clusters.
%\allowdisplaybreaks

\begin{lemma}~\label{lemma:far_cluster}
%If we sample a point $p \in M_{i}$ uniformly at random, then the following identity holds:
Let $p$ be a point sampled uniformly at random from $M_i$. Then the following bound holds:
\[
    \mathbb{E}[\Phi(t(p),C_{i})] \leq \left( 3^{\ell}+\frac{\veps}{4} \right) \cdot \Delta(C_{i})
\]
\end{lemma}
\vspace*{1mm}
\begin{proof}
\hspace{10mm}
$
\begin{aligned}[t]
    \mathbb{E}[\Phi(t(p),C_{i})] &= \frac{1}{|C_{i}|} \cdot \left( \sum_{p \in M_{i}} \Phi(t(p),C_{i}) \right) \\
    &= \frac{1}{|C_{i}|} \cdot \left( \sum_{x' \in C_{i}^{n}} \Phi(t(c(x')),C_{i}) + \sum_{x' \in C_{i}^{f}} \Phi(t(x'),C_{i}) \right) \\
\end{aligned}
$
~\\
\vspace*{3mm}

\noindent Let us evaluate these two terms separately.  
\begin{enumerate}
    \item The first term:
    \begin{align*}
    \hspace{3mm}\sum_{x' \in C_{i}^{n}} \Phi(t(c(x')),C_{i}) &= \sum_{x' \in C_{i}^{n}} \sum_{x \in C_{i}} d^{\ell}(x,t(c(x'))) \\
    &\leq \sum_{x' \in C_{i}^{n}} \sum_{x \in C_{i}} \left( d(x,x') + d(x',c(x')) + d(c(x'),t(c(x'))) \right)^{\ell}, \\
    & \hspace{70mm}
    \textrm{by triangle-inequality}\\
    &\leq \sum_{x' \in C_{i}^{n}} \sum_{x \in C_{i}} \left( d(x,x') + d(x',c(x')) + d(c(x'),f_{i}^{*}) \right)^{\ell}, \\
    &\hspace{70mm} \textrm{by the defn. of $t(c(x'))$} \\
    &\leq \sum_{x' \in C_{i}^{n}} \sum_{x \in C_{i}} \left( d(x,x') + d(x',c(x')) + d(x',c(x')) + d(x',f_{i}^{*}) \right)^{\ell}, \\
    & \hspace{70mm}
    \textrm{by triangle-inequality}\\
    &\leq \sum_{x' \in C_{i}^{n}} \sum_{x \in C_{i}} \left( \left(1+\frac{1}{\delta} \right)^{\ell} \cdot (d(x,x') + d(x',f_{i}^{*}) )^{\ell} + 2^{\ell}(1+\delta)^{\ell} \cdot d^{\ell}(x',c(x')) \right),\\
    &\hspace{70mm} \textrm{by Fact~\ref{fact:trade_off}}
    \end{align*}
    
    \item The second term:
    \begin{align*}
    \hspace{3mm}\sum_{x' \in C_{i}^{f}} \Phi(t(x'),C_{i}) &= \sum_{x' \in C_{i}^{f}} \sum_{x \in C_{i}} (d(x,x') + d(x',t(x')))^{\ell}, \quad 
    \textrm{by triangle-inequality}\\
    &\leq \sum_{x' \in C_{i}^{f}} \sum_{x \in C_{i}} (d(x,x') + d(x',f_{i}^{*}))^{\ell}, \quad \quad
    \textrm{by the defn. of $t(x')$}
    \end{align*}
\end{enumerate}

\noindent On combining the two terms we get the following expression:
\begin{align*}
    \mathbb{E}[\Phi(t(p),C_{i})]
    &\leq  \left(1+\frac{1}{\delta} \right)^{\ell} \cdot \frac{1}{|C_{i}|} \cdot \left(  \sum_{x' \in C_{i}} \sum_{x \in C_{i}} (d(x,x') + d(x',f_{i}^{*}))^{\ell} \right) \\
    &\hspace{40mm} + (1+\delta)^{\ell} \cdot \frac{1}{|C_{i}|} \cdot \left( \sum_{x' \in C_{i}^{n}} |C_{i}| \cdot  2^{\ell}d^{\ell}(x',c(x')) \right)\\
\end{align*}

\noindent By lemma~\ref{lemma:exp}, we have $\mathbb{E}[\Phi(t(x),C_{i})] \leq \frac{1}{|C_{i}|} \cdot \Big( \sum_{x' \in C_{i}} \sum_{x \in C_{i}} (d(x,x') + d(x',f_{i}^{*}) )^{\ell} \Big) \leq 3^{\ell} \cdot \Delta(C_{i})$. Thus, we get:
\begin{align*}
    \mathbb{E}[\Phi(t(p),C_{i})] &\leq \left(1+\frac{1}{\delta} \right)^{\ell} \cdot 3^{\ell} \cdot \Delta(C_{i}) + (1+\delta)^{\ell} \cdot \frac{1}{|C_{i}|} \cdot \left( \sum_{x' \in C_{i}^{n}} |C_{i}| \cdot  2^{\ell}d^{\ell}(x',c(x')) \right)\\
    &= \left(1+\frac{1}{\delta} \right)^{\ell} \cdot 3^{\ell} \cdot \Delta(C_{i}) + 2^{\ell}(1+\delta)^{\ell} \cdot \Phi(F,C_{i}^{n})\\
    &= \left(1+\frac{\veps}{\ell \cdot  3^{\ell+3}} \right)^{\ell} \cdot 3^{\ell} \cdot \Delta(C_{i}) + 2^{\ell}(1+\delta)^{\ell} \cdot \Phi(F,C_{i}^{n}), \hspace{8mm} \textrm{by substituting $\delta = \frac{\ell \cdot 3^{\ell+3}}{\veps}$}\\
    &\leq \left(1+2 \ell \cdot \frac{\veps}{\ell \cdot 3^{\ell+3}} \right) \cdot 3^{\ell} \cdot \Delta(C_{i}) + 2^{\ell}(1+\delta)^{\ell} \cdot \Phi(F,C_{i}^{n}), \hspace{3mm} \textrm{by Fact~\ref{fact:bin_approx}} \\
    &= \left( 3^{\ell} + \frac{\veps}{8} \right) \cdot \Delta(C_{i}) + 2^{\ell}(1+\delta)^{\ell} \cdot \Phi(F,C_{i}^{n}), \\
    &= \left( 3^{\ell} + \frac{\veps}{8} \right) \cdot \Delta(C_{i}) + 2^{\ell} \cdot (2\delta)^{\ell} \cdot \Phi(F,C_{i}^{n}), \hspace{10mm} \textrm{$\because 1 \leq \delta$, for $\veps \leq 1$} \\
    &\leq \left( 3^{\ell} + \frac{\veps}{8} \right) \cdot \Delta(C_{i}) + \frac{\veps}{8} \cdot \Delta(C_{i}), \hspace{27mm} \textrm{by lemma~\ref{lemma:radius}}\\
    &= \left( 3^{\ell} + \frac{\veps}{4} \right) \cdot \Delta(C_{i})\\
\end{align*}
This completes the proof of the lemma.
\end{proof}

\noindent We obtain the following corollary
using the Markov's inequality.
\begin{corollary}\label{corollary:high_cost}
%If we sample a point $p \in M_{i}$, uniformly at random, then for $\veps \leq 1$:
For any $0 < \veps \leq 1$ and point $p$ sampled uniformly at random from $M_i$, we have:
\[
\emph{\textbf{Pr}}\left[ \Phi(t(p),C_{i}) \leq \left( 3^{\ell}+\frac{\veps}{2} \right) \cdot \Delta(C_{i}) \right] > \frac{\veps}{3^{\ell+2}}.
\]
\end{corollary}

\noindent Let us call a point $p$, a \emph{good} point if $t(p)$ gives $\Big( 3^{\ell}+\frac{\veps}{2} \Big)$-approximation for $C_{i}$.

\begin{corollary}
There are at least $\left \lceil \fdfrac{\veps \cdot |C_{i}|}{3^{\ell+2}} \right \rceil$ good points in $M_{i}$.
\end{corollary}

\noindent Now our goal is to obtain one such good point from $M_{i}$.  If $F$ contains any good point then we are done, since the algorithm adds the entire set $F$ to $M$. As a result, Property I is satisfied for the cluster $C_{i}$.
On the other hand, if $F$ does not contain any good point then there is no good point in $M_{i}^{n}$ as well. It simply means that all good points are present in $C_{i}^{f}$. To sample these good points we use the $D^{\ell}$-sampling technique. Let $G \subseteq C_{i}^{f}$ denote the set of good points. Then we have $|G| \geq \left \lceil \fdfrac{\veps \cdot |C_{i}|}{3^{\ell+2}} \right \rceil$. Using this fact, we will prove the following lemma.

\begin{lemma}\label{lemma:d2_sample}
%If we $D^{\ell}$-sample a point $x \in C$, then for any point $p \in C_{i}^{f}$, 
For any point $p \in C_{i}^{f}$ and any $D^{\ell}$-sampled point $x \in C$, we have:
$\emph{\textbf{Pr}}[x = p] \geq \fdfrac{\veps}{ \alpha \, \beta \, \gamma \, k \, |C_{i}|} = \tau$ and  
$\emph{\textbf{Pr}}\left[ \Phi(t(x),C_{i}) \leq \Big( 3^{\ell}+\frac{\veps}{2} \Big) \cdot \Delta(C_{i} ) \right] \geq \fdfrac{\veps^{2}}{ \alpha \, \beta \, \gamma \, k \cdot 3^{\ell+2}}$
\end{lemma}
\begin{proof}
For any point $p \in C_{i}^{f}$,
\[
\textbf{Pr}[x = p] = \frac{d^{\ell}(p,c(p))}{\Phi(F,C)} \geq \frac{R^{\ell}}{\Phi(F,C)} = \frac{1}{\beta \, |C_{i}|} \cdot \frac{\Phi(F,C_{i})}{\Phi(F,C)} \geq \frac{\veps}{\alpha \, \beta \, \gamma \, k \, |C_{i}|} 
\]
Let $Z$ denote an indicator random variable, such that $Z=1$ if $\Phi(t(x),C_{i}) \leq \Big( 3^{\ell}+\frac{\veps}{2} \Big) \cdot \Delta(C_{i} )$ and 0 otherwise.
\begin{align*}
\textbf{Pr}[\, Z = 1 \,] &\geq \sum_{p \in G} \, \textbf{Pr}[\, x = p \,] \cdot \textbf{Pr}[\left. Z = 1 \;\middle|\; x = p \right.] \\ 
&\geq \sum_{p \in G} \, \frac{\veps}{\alpha \, \beta \, \gamma \, k \, |C_{i}|} \cdot 1 \\
&= |G| \cdot \frac{\veps}{\alpha \, \beta \, \gamma \, k \, |C_{i}|} 
\geq \frac{\veps^{2}}{ \alpha \, \beta \, \gamma \, k \cdot 3^{\ell+2}}
\end{align*}
This completes the proof of the lemma.
\end{proof}

\noindent The above lemma states that, every point in $C_i^{f}$ has a sampling probability of at least $\tau$. 
%Equivalently, it can be said that $D^{\ell}$-sampling gives almost uniform samples from $C_i^{f}$. 
Moreover, a sampled point gives $(3^{\ell} + \frac{\veps}{2})$-approximation for $C_{i}$ with probability at least $\frac{\veps^{2}}{(\alpha \, \beta \, \gamma \, k \cdot 3^{\ell+2})}$.
To boost this probability, we sample $\eta \coloneqq \fdfrac{\, \alpha \, \beta \, \gamma \, k \cdot 3^{\ell+2}}{\veps^{2}}$ points independently from $C$ using $D^{\ell}$-sampling (see line (4) of the algorithm). It follows that, with probability at least $1/2$, there is a point in the sampled set that gives $(3^{\ell} + \frac{\veps}{2})$-approximation for $C_{i}$. Hence, Property-I is satisfied for $C_{i}$. 
Also note that, in line(4) of the algorithm, we sample $\eta \cdot k$ points, i.e., $\eta$ points corresponding to each cluster. Hence, Property-I  holds for every cluster in $H$. 

Since Property-I is satisfied for every cluster in $W$ and $H$, we can finally claim that $T_{s} = \{t(s_{1}), t(s_{2}), \dotsc, t(s_{k})\}$ is a $\left( 3^{\ell} + \veps \right)$-approximation for $\C$ with probability at least $\frac{1}{2^{k}}$. 
However, as described earlier, $T_{s}$ could be a soft center-set since $t(s_{i})$ can be same as $t(s_{j})$ for some $i \neq j$. 
To obtain a hard center-set, we make use of line (8) of the algorithm. 
In line (8), the algorithm pulls out the $k$ closest points from $L$ instead of just one. 
Note that it is not necessary to open a facility at a closest location in $L$. 
Rather, we can open a facility at any location $f$ in $L$, that is at least as close to $s_{i}$ as $f_{i}^{*}$, i.e., $d(s_{i},f) \leq d(s_{i},f_{i}^{*})$.

\noindent Let $T(s_{i})$ denote a set of $k$ closest facility location for $s_{i}$. We show that there is a hard center-set $T_{h} \subset \displaystyle \cup_{i} T(s_{i}) $, such that $T_{h} \coloneqq \{f_{1},\dotsc,f_{k}\}$ and $d(s_{i},f_{i}) \leq d(s_{i},f_{i}^{*})$ for every $1 \leq i \leq k$.
We define $T_{h}$ using the following simple subroutine:
\begin{framed}
{\tt FindFacilities}\\
\hspace*{0.1in} - $T_h \leftarrow \emptyset$\\
\hspace*{0.1in} - For $i \in \{1, ..., k\}$:\\
\hspace*{0.3in} - if ($f_i^* \in T(s_i)$) $T_h \leftarrow T_h \cup \{f_i^*\}$\\
\hspace*{0.3in} - else \\
\hspace*{0.5in} - Let $f \in T(s_i)$ be any facility such that $f$ is not in $T_h$\\
\hspace*{0.5in} - $T_h \leftarrow T_h \cup \{f\}$
\end{framed}

%\begin{enumerate}
%    \item For any $1 \leq i \leq k$, if $f_{i}^{*} \in T(s_{i})$, then $f_{i} \gets f_{i}^{*}$.
%    \item For the remaining indices, assign any facility of $T(s_{i})$ to $f_{i}$, as long as it is not present in $T_{h}$.
%\end{enumerate}
%Since size of each $T(s_{i})$ is exactly $k$, We can always carry out the second step, till $k$ facilities are added to $T_{h}$.

\begin{lemma}
$T_{h} = \{ f_{1},f_{2}, \dotsc,f_{k} \}$ contains exactly $k$ different facilities such that for every $1 \leq i \leq k$, we have $d(s_{i},f_{i}) \leq d(s_{i},f_{i}^{*})$.
\end{lemma}
\begin{proof}
First, let us show that all facilities in $T_{s}$ are different.
Since, $f_{i}^{*}$ is different for different clusters, the if statement adds facilities in $T_{h}$ that are different. In else part, we only add a facility to $T_{h}$ that is not present in $T_{h}$. Thus the else statement also adds facilities in $T_{h}$ that are different. 

Now, let us prove the second property, i.e., $d(s_{i},f_{i}) \leq d(s_{i},f_{i}^{*})$ for every $1 \leq i \leq k$. The property is trivially true for the facilities added in the if statement. 
Now, for the facilities added in the second step we know that $T(s_{i})$ does not contain $f_{i}^{*}$. Since, $T(s_{i})$ is a set of $k$-closest facility locations, we can say that for any facility location $f$ in $T(s_{i})$, $d(s_{i},f) \leq d(s_{i},f_{i}^{*})$. Thus any facility added in the else statement has $d(s_{i},f) \leq d(s_{i},f_{i}^{*})$. This completes the proof.
\end{proof}
\noindent Thus $T_{h} \in \mathcal{L}$ is a hard center-set, which gives the $(3^{\ell}+\veps)$-approximation for the problem. This completes the analysis of the algorithm.

Now, suppose we are given the flexibility to open a facility at a client location. In other words, suppose it is given that $C \subseteq L$. For this case, we can directly open the facilities at the locations $\{s_{1},s_{2},\dotsc,s_{k}\}$ instead of $t(s_{i})$'s, and we would not need lines (7) and (8) of the algorithm. 
Further, we can show that lemma~\ref{lemma:close_cluster} and~\ref{lemma:far_cluster} would give $(2^{\ell} + \veps)$-approximation for this special case. 
However, please note that $\{s_{1},s_{2},\dotsc,s_{k}\}$ is still a soft center-set. 
To obtain a hard center-set we do need to consider the $k$-closest facility locations for a point $s_{i}$. 
In that case, lemma~\ref{lemma:close_cluster} and~\ref{lemma:far_cluster} cannot provide $(2^{\ell}+\veps)$-approximation. Therefore, we make some slight changes in the analysis of lemma~\ref{lemma:close_cluster} and~\ref{lemma:far_cluster} to get a $(2^{\ell}+\veps)$-approximation. 
We discuss those details in the following sub-section.

\subsection{Analysis for \texorpdfstring{$\mathbf{C \subseteq L}$}{}}
The algorithm \LK($C,L,k,d,\ell,\veps$), gives better approximation guarantees if it is allowed to open a facility at a client location, i.e., $C \subseteq L$.
We will show the following result.
\begin{theorem}
Let $0 < \veps \leq 1$.
Let $(C,L,k,d,\ell)$ be any $k$-service instance with $C \subseteq L$ and let $\mathcal{C} = \{C_{1}, C_{2},\dotsc,C_{k}\}$ be any arbitrary clustering of the client set.
%For a $k$-service instance $(C,L,k,d,\ell)$ and for an arbitrary clustering $\mathcal{C} = \{C_{1}, C_{2},\dotsc,C_{k}\}$ of the client set, 
The algorithm ~ \LK($C,L,k,d,\ell,\veps$), with probability at least $1/2$, outputs a list $\mathcal{L}$ of size $(k /\veps)^{O(k \, \ell^{\,2})}$, such that there is a $k$ center set $S \in \mathcal{L}$ in the list such that
\[
    \Psi (S,\C) \leq (2^{\ell} + \veps) \cdot \Psi^{*}(\C).
\]
Moreover, the running time of the algorithm is $O \left( n \cdot (k /\veps)^{O(k \, \ell^{\,2})} \right)$.
\end{theorem}

\noindent Since we are dealing with a special case, all the previous lemmas (i.e., Lemma~\ref{lemma:close_cluster},~\ref{lemma:radius},~\ref{lemma:far_cluster}, and~\ref{lemma:d2_sample}) are valid here as well. 
Here we will obtain improved versions of  Lemmas~\ref{lemma:close_cluster} and~\ref{lemma:far_cluster} and consequently obtain better approximation guarantee ($2^{\ell}$ instead of $3^{\ell}$). 
We show the following property for set $M$.
\begin{quote}{\textbf{Property}-\boldmath{$\rm II$}:}\label{property:1}
For any cluster $C_{i} \in \{C_{1}, C_{2}, \dotsc, C_{k}\}$,
there is a point $s_{i}$ in $M$ such that with probability at least $1/2$, following holds:
\[
    \Phi(u_i(s_{i}),C_{i}) \leq 
\begin{cases}
    \Big( 2^{\ell} + \fdfrac{\veps}{2} \Big) \cdot \Delta(C_{i}) + \fdfrac{ \veps}{2^{\ell+1} \,k} \cdot OPT(C,C),& \textrm{if } C_{i} \in W \\[7pt]
    \Big( 2^{\ell}+\fdfrac{\veps}{2}
    \Big) \cdot \Delta(C_{i}) ,& \textrm{if } C_{i} \in H.
\end{cases}
\]
Here, for any point $x \in C \cup L$, $u_i(x)$ denotes a location that is as close to $x$ as its \emph{closest} point in $C_{i}$. In other words, if $p = \argmin_{y \in C_{i}} \{ d(y,x) \}$ is the closest location in $C_{i}$, then $d(x,u_i(x)) \leq d(x,p)$.
\end{quote}

\noindent By Property-II, we claim that $\{u_1(s_{1}), u_2(s_{2}), \dotsc, u_k(s_{k}) \}$ is $\Big( 2^{\ell} 
+ \frac{\veps}{2} \Big)$-approximation for $\mathcal{C}$. 
The reasoning for this claim is the same as the one we provided in the previous subsection. Moreover, $s_{i}$ is always a possible candidate for $u(s_{i})$ and it is added to the set $T$ in line (8) of the algorithm. Therefore, $\mathcal{L}$ contains a set that is $\Big( 2^{\ell} 
+ \frac{\veps}{2} \Big)$-approximation for $\mathcal{C}$. However, this set is only a soft center-set. Later we will show that there is a hard center-set in $\mathcal{L}$ that gives the desired approximation for the problem. For now, let us prove Property-II for the target clusters.\\

\noindent \textbf{Case 1: \boldmath{$\Phi(F,C_{i}) \leq \fdfrac{\veps}{\alpha \, \gamma \, k} \cdot \Phi(F,C)$}}
~\\~\\
Recall that we defined a multi-set $M_{i} \coloneqq \{ c(x) \mid x \in C_{i} \}$, where $c(x)$ denotes a location in $F$ that is closest to $x$. We show the following result.
\begin{lemma}~\label{lemma:close_cluster1}
Let $p$ be a point sampled uniformly at random from $M_i$. Then the following bound holds:
%If we sample a point $p \in M_{i}$ uniformly at random, then following holds:
\[
    \mathbb{E}[\Phi(u_i(p),C_{i})] \leq \left( 2^{\ell} + \frac{\veps}{2} \right) \cdot \Delta(C_{i}) + \frac{ \veps}{2^{\ell+1} \,k} \cdot OPT(C,C)
\]
\end{lemma}

\begin{proof} 
The proof follows from the following sequence of inequalities.
\begin{align*}
    \mathbb{E}[\Phi(u_i(p),C_{i})] &= \frac{1}{|C_{i}|} \cdot \left( \sum_{p \in M_{i}} \Phi(u_i(p),C_{i}) \right) \\
    &= \frac{1}{|C_{i}|} \cdot \left(  \sum_{p \in M_{i}} \sum_{x \in C_{i}} d^{\ell}(x,u_i(p))  \right) \\
    &= \frac{1}{|C_{i}|} \cdot \left( \sum_{x' \in C_{i}} \sum_{x \in C_{i}} d^{\ell}(x,u_i(c(x'))) \right) \\
    &\leq \frac{1}{|C_{i}|} \cdot \left( \sum_{x' \in C_{i}} \sum_{x \in C_{i}} \left(~ d(x,x') + d(x',c(x')) + d(c(x'),u_i(c(x')))~ \right)^{\ell} \right), \\
    &\qquad \qquad \qquad \qquad \qquad \qquad \qquad \qquad \textrm{by triangle inequality} \\
    &\leq \frac{1}{|C_{i}|} \cdot \left( \sum_{x' \in C_{i}} \sum_{x \in C_{i}} \left( d(x,x') + d(x',c(x')) + d(c(x'),x' ) \right)^{\ell} \right), \\
    &\qquad \qquad \qquad \qquad \qquad \qquad \qquad \qquad \textrm{by the defn. of $u_i(c(x'))$} \\
    &= \frac{1}{|C_{i}|} \cdot \left( \sum_{x' \in C_{i}} \sum_{x \in C_{i}} \left( d(x,x') + 2d(x',c(x')) \right)^{\ell} \right) \\
\end{align*}
\noindent Let us use Fact~\ref{fact:trade_off} by setting $b = d(x,x')$ and $a = 2 \cdot d(x',c(x'))$. We obtain the following expression:
\begin{align*}
     \hspace{22mm} &\leq \frac{1}{|C_{i}|} \cdot \left( \sum_{x' \in C_{i}} \sum_{x \in C_{i}} \left( \left(1+\frac{1}{\delta} \right)^{\ell} \cdot d^{\ell}(x,x') + (1+\delta)^{\ell} \cdot 2^{\ell}d^{\ell}(x',c(x')) \right) \right), \\
    &\hspace{100mm} \textrm{for any $\delta > 0 $} \\
    &=  \left(1+\frac{1}{\delta} \right)^{\ell} \cdot \frac{1}{|C_{i}|} \cdot \left(  \sum_{x' \in C_{i}} \sum_{x \in C_{i}} d^{\ell}(x,x') \right) +
    (1+\delta)^{\ell} \cdot \frac{1}{|C_{i}|} \cdot \left( \sum_{x' \in C_{i}} |C_{i}| \cdot  2^{\ell}d^{\ell}(x',c(x')) \right) \\
    &= \left(1+\frac{1}{\delta} \right)^{\ell} \cdot \mathbb{E}[\Phi(x,C_{i})] + 2^{\ell}(1+\delta)^{\ell} \cdot \Phi(F,C_{i})\\
    &\leq \left(1+\frac{1}{\delta} \right)^{\ell} \cdot 2^{\ell} \cdot \Delta(C_{i}) + 2^{\ell}(1+\delta)^{\ell} \cdot \Phi(F,C_{i}), \hspace{17mm} \textrm{by lemma~\ref{lemma:exp2}} \\
    &= \left(1+\frac{\veps}{\ell \cdot  3^{\ell+2}} \right)^{\ell} \cdot 2^{\ell} \cdot \Delta(C_{i}) + 2^{\ell}(1+\delta)^{\ell} \cdot \Phi(F,C_{i}), \hspace{9mm} \textrm{by substituting $\delta = \frac{\ell \cdot 3^{\ell+2}}{\veps}$}\\
    &\leq \left(1+2 \ell \cdot \frac{\veps}{\ell \cdot 3^{\ell+2}} \right) \cdot 2^{\ell} \cdot \Delta(C_{i}) + 2^{\ell}(1+\delta)^{\ell} \cdot \Phi(F,C_{i}), \quad \textrm{by Fact~\ref{fact:bin_approx}} \\
    &= \left( 2^{\ell} + \frac{\veps}{2} \right) \cdot \Delta(C_{i}) + 2^{\ell}(1+\delta)^{\ell} \cdot \Phi(F,C_{i}), \\
    &\leq \left( 2^{\ell} + \frac{\veps}{2} \right) \cdot \Delta(C_{i}) + 2^{\ell} \cdot (2\delta)^{\ell} \cdot \Phi(F,C_{i}), \hspace{10mm} \textrm{$\because 1 \leq \delta$, for $\veps \leq 1$}\\
    &\leq \left( 2^{\ell} + \frac{\veps}{2} \right) \cdot \Delta(C_{i}) + \frac{4^{\ell} \cdot \delta^{\ell} \cdot \veps}{\alpha \, \gamma \, k} \cdot \Phi(F,C), \hspace{10mm} \textrm{$\because \Phi(F,C_{i}) \leq \frac{\veps}{\alpha \, \gamma \, k} \cdot \Phi(F,C)$}\\
    &\leq \left( 2^{\ell} + \frac{\veps}{2} \right) \cdot \Delta(C_{i}) + \frac{\veps}{2^{\ell+1} \, \alpha \, k} \cdot \Phi(F,C), \hspace{11mm} \textrm{$\because \gamma = \frac{\ell^{\ell} \cdot 3^{\ell^{2} + 5\ell + 1}}{\veps^{\ell}}$}\\
    &\leq \left( 2^{\ell} + \frac{\veps}{2} \right) \cdot \Delta(C_{i}) + \frac{\veps}{2^{\ell+1} \,k} \cdot OPT(L,C), \hspace{8mm} \textrm{$\because \Phi(F,C) \leq \alpha \cdot OPT(C,C)$} \\
\end{align*}
This completes the proof of the lemma.
\end{proof}
\begin{onehalfspace}
\noindent By the above lemma, we can claim that there is a point $p$ in $M_{i}$ such that $\Phi(u_i(p),C_{i}) \leq \left( 2^{\ell} + \fdfrac{\veps}{2} \right) \cdot \Delta(C_{i}) + \fdfrac{\veps}{2^{\ell+1} \, k} \cdot OPT(L,C)$. 
Since $M_{i}$ is only composed of the points from $F$ and $F$ is contained in $M$, we can say that Property- II is satisfied for every cluster $C_{i} \in W$.
Now, let us analyze the high cost clusters.
\end{onehalfspace}

~\\ \\
\textbf{Case 2: \boldmath{$\Phi(F,C_{i}) > \fdfrac{\veps}{\alpha \, \gamma \, k} \cdot \Phi(F,C)$}}. ~\\

\noindent We will reuse $C_i^n$ and $C_i^f$ as defined earlier.
As before, let us define a multi-set $M_{i}^{n} \coloneqq \{ c(x) \mid x \in C_{i}^{n} \}$ and a multi-set $M_{i} \coloneqq C_{i}^{f} \cup M_{i}^{n}$. We show the following result.
%\allowdisplaybreaks

\begin{lemma}~\label{lemma:far_cluster2}
Let $p$ be a point sampled uniformly at random from $M_i$. Then the following bound holds:
%If we sample a point $p \in M_{i}$ uniformly at random, then following holds:
\[
    \mathbb{E}[\Phi(u_i(p),C_{i})] \leq \left( 2^{\ell}+\frac{\veps}{4} \right) \cdot \Delta(C_{i})
\]
\end{lemma}
\vspace*{2mm}
\begin{proof}
\hspace*{10mm}
$
\begin{aligned}[t]
    \mathbb{E}[\Phi(u_i(p),C_{i})] &= \frac{1}{|C_{i}|} \cdot \left( \sum_{p \in M_{i}} \Phi(u_i(p),C_{i}) \right) \\
    &= \frac{1}{|C_{i}|} \cdot \left( \sum_{x' \in C_{i}^{n}} \Phi(u_i(c(x')),C_{i}) + \sum_{x' \in C_{i}^{f}} \Phi(u_i(x'),C_{i}) \right) \\
\end{aligned}$
\\~\\

\noindent Let us evaluate these two terms separately.
\vspace*{2mm}
\begin{enumerate}
    \item The first term:
    \hspace{2mm} \begin{align*}
    \sum_{x' \in C_{i}^{n}} \Phi(u_i(c(x')),C_{i}) &= \sum_{x' \in C_{i}^{n}} \sum_{x \in C_{i}} d^{\ell}(x,u_i(c(x'))) \\
    &\leq \sum_{x' \in C_{i}^{n}} \sum_{x \in C_{i}} \left( d(x,x') + d(x',c(x')) + d(c(x'),u_i(c(x')))) \right)^{\ell}, \\
    & \hspace{70mm}
    \textrm{by triangle-inequality}\\
    &\leq \sum_{x' \in C_{i}^{n}} \sum_{x \in C_{i}} \left( d(x,x') + d(x',c(x')) + d(c(x'),x') ) \right)^{\ell}, \\
    &\hspace{70mm} \textrm{by the defn. of $u_i(c(x'))$}\\
    &= \sum_{x' \in C_{i}^{n}} \sum_{x \in C_{i}} \left( d(x,x') + 2d(x',c(x')) ) \right)^{\ell}, \\
    &\leq \sum_{x' \in C_{i}^{n}} \sum_{x \in C_{i}} \left( \left(1+\frac{1}{\delta} \right)^{\ell} \cdot d^{\ell}(x,x') + 2^{\ell}(1+\delta)^{\ell} \cdot d^{\ell}(x',c(x')) \right)\\
    &\hspace{70mm} \textrm{by Fact~\ref{fact:trade_off}}
    \end{align*}
    \vspace*{1mm}
    
    \item The second term:
    \begin{align*}
    \sum_{x' \in C_{i}^{f}} \Phi(u_i(x'),C_{i}) &= \sum_{x' \in C_{i}^{f}} \sum_{x \in C_{i}} (d(x,x') + d(x',u_i(x')))^{\ell}, \quad
    \textrm{by triangle-inequality}\\
    &= \sum_{x' \in C_{i}^{f}} \sum_{x \in C_{i}} d^{\ell}(x,x'), \hspace{29mm}
    \textrm{by the defn. of $u_i(x')$}
    \end{align*}
\end{enumerate}

\noindent Let us now combine the two terms. We get the following expression:
\begin{align*}
    \mathbb{E}[\Phi(u_i(p),C_{i})]
    &\leq  \left(1+\frac{1}{\delta} \right)^{\ell} \cdot \frac{1}{|C_{i}|} \cdot \left(  \sum_{x' \in C_{i}} \sum_{x \in C_{i}} d^{\ell}(x,x') \right)
    + (1+\delta)^{\ell} \cdot \frac{1}{|C_{i}|} \cdot \left( \sum_{x' \in C_{i}^{n}} |C_{i}| \cdot  2^{\ell}d^{\ell}(x',c(x')) \right) \\
    &= \left(1+\frac{1}{\delta} \right)^{\ell} \cdot \mathbb{E}[\Phi(x,C_{i})] + 2^{\ell}(1+\delta)^{\ell} \cdot \Phi(F,C_{i}^{n})\\
    &\leq \left(1+\frac{1}{\delta} \right)^{\ell} \cdot 2^{\ell} \cdot \Delta(C_{i}) + 2^{\ell}(1+\delta)^{\ell} \cdot \Phi(F,C_{i}^{n}), \hspace{16mm} \textrm{by lemma~\ref{lemma:exp}}\\
    &= \left(1+\frac{\veps}{\ell \cdot  3^{\ell+3}} \right)^{\ell} \cdot 2^{\ell} \cdot \Delta(C_{i}) + 2^{\ell}(1+\delta)^{\ell} \cdot \Phi(F,C_{i}^{n}), \hspace{8mm} \textrm{by substituting $\delta = \frac{\ell \cdot 3^{\ell+3}}{\veps}$}\\
    &\leq \left(1+2 \ell \cdot \frac{\veps}{\ell \cdot 3^{\ell+3}} \right) \cdot 2^{\ell} \cdot \Delta(C_{i}) + 2^{\ell}(1+\delta)^{\ell} \cdot \Phi(F,C_{i}^{n}), \hspace{3mm} \textrm{by Fact~\ref{fact:bin_approx}} \\
    &= \left( 2^{\ell} + \frac{\veps}{8} \right) \cdot \Delta(C_{i}) + 2^{\ell}(1+\delta)^{\ell} \cdot \Phi(F,C_{i}^{n}), \\
    &= \left( 2^{\ell} + \frac{\veps}{8} \right) \cdot \Delta(C_{i}) + 2^{\ell} \cdot (2\delta)^{\ell} \cdot \Phi(F,C_{i}^{n}), \hspace{8mm} \because 1 \leq \delta, \textrm { for $\veps \leq 1$}\\
    &\leq \left( 2^{\ell} + \frac{\veps}{8} \right) \cdot \Delta(C_{i}) + \frac{\veps}{8} \cdot \Delta(C_{i}), \hspace{25mm} \textrm{by lemma~\ref{lemma:radius}}\\
    &= \left( 2^{\ell} + \frac{\veps}{4} \right) \cdot \Delta(C_{i})
\end{align*}
This completes the proof of the lemma.
\end{proof}

\noindent We obtain the following corollary
using the Markov's inequality.
\begin{corollary}\label{corollary:high_cost2}
If we sample a point $p \in M_{i}$, uniformly at random, then for $\veps \leq 1$:
\[
\emph{\textbf{Pr}}[\Phi(u_i(p),C_{i}) \leq \left( 2^{\ell}+\frac{\veps}{2} \right) \cdot \Delta(C_{i})] \geq \frac{\veps}{2^{\ell+2}} > \frac{\veps}{3^{\ell+2}} 
\]
\end{corollary}

\noindent Since Corollary~\ref{corollary:high_cost} is similar to Corollary~\ref{corollary:high_cost2}, all arguments made in the previous subsection are valid here as well. Thus we can claim that, with probability at least $1/2$ there is a point $x$ in the sampled set $M$ such that $u_i(x)$ gives $\Big(2^{\ell}+\frac{\veps}{2} \Big)$-approximation for $C_{i}$. In other words, Property-II is satisfied for every cluster in $H$.

Now let us show that there is a hard center-set $T_{h} = \{f_{1},f_{2},\dotsc,f_{k}\}$ in the list, which gives the desired approximation. 
Our argument is based on the fact that we can open a facility at any location $u_i(s_{i})$, which is at least as close to $s_{i}$ as the closest location in $C_{i}$ (see Property-II). 
Let $p_{i}$ denote a location in $C_{i}$ that is closest to $s_{i}$ i.e. $p_{i} = \argmin_{x \in C_{i}} d(x,s_{i})$. 
We will show that $d(s_{i},f_{i}) \leq d(s_{i},p_{i})$ for every $1 \leq i \leq k$. 
The following analysis is very similar to the analysis that we did in the previous section. 
The only difference is that $f_{i}^{*}$ is replaced with $p_{i}$.
\vspace*{2mm}

\noindent Let $T(s_{i})$ denote a set of $k$ closest facility locations for a sampled point $s_{i} \in M$ (see line(8) of the algorithm). We define $T_{h} = \{ f_{1},f_{2}, \dotsc,f_{k} \}$ using the following simple subroutine:
\begin{framed}
{\tt FindFacilities}\\
\hspace*{0.1in} - $T_h \leftarrow \emptyset$\\
\hspace*{0.1in} - For $i \in \{1, ..., k\}$\\
\hspace*{0.3in} - if ($p_i \in T(s_i)$) $T_h \leftarrow T_h \cup \{p_i\}$\\
\hspace*{0.3in} - else\\
\hspace*{0.5in} - Let $f \in T(s_i)$ be any facility such that $f$ is not in $T_h$\\
\hspace*{0.5in} - $T_h \leftarrow T_h \cup \{f\}$
\end{framed}

\noindent Since size of each $T(s_{i})$ is exactly $k$, $T_h$ will contain $k$ facilities at the end of the for-loop.
\vspace*{2mm}
\begin{lemma}
The subroutine picks a set $T_{h} = \{ f_{1},f_{2}, \dotsc,f_{k} \}$ of $k$ distinct facilities such that 
for every $1 \leq i \leq k$, $d(s_{i},f_{i}) \leq d(s_{i},p_{i})$.
\end{lemma}
\begin{proof}
First let us show that the facilities in $T_{h}$ are distinct.
Since $p_i$'s are different for different clusters, the if statement adds distinct facilities to $T_h$. 
Moreover, in the else statement, we only add a facility to $T_{h}$ if it is not already present in $T_{h}$. This also ensures the facilities added in $T_{h}$ are distinct.

Now, let us prove the second property, i.e., $d(s_{i},f_{i}) \leq d(s_{i},p_{i})$ for every $1 \leq i \leq k$. The property is trivially true for the facilities added in the if statement. 
Now, for the facilities added in the else statement, it is known that $T(s_{i})$ does not contain $p_{i}$. 
Since, $T(s_{i})$ is a set of $k$-closest facility locations, for any facility facility $f$ in $T(s_{i})$, we have $d(s_{i},f) \leq d(s_{i},p_{i})$. 
Thus any facility added in step 2 has $d(s_{i},f) \leq d(s_{i},p_{i})$. This completes the proof.
\end{proof}

\noindent Thus $T_{h} \in \mathcal{L}$ is a hard center-set, which gives $(2^{\ell}+\veps)$-approximation for the problem. This completes the analysis of the algorithm when $C \subseteq L$.

\subsection{Algorithm for Outlier \texorpdfstring{$k$}{}-Service }\label{section:outlier_list_k_service}

The outlier $k$-service problem does not fit the framework of constrained $k$-service problem since a clustering is defined over the set $C \setminus Z$ instead of $C$. 
Due to this we can not use the list $k$-service algorithm. 
To overcome this issue, we modify the algorithm \LK, so that it gives a good center-set for any clustering defined over the set $C \setminus Z$. 
We do this modification only in the first step of the algorithm. 
In the first step, we run a $\alpha$-approximation algorithm for the $(k+m)$-means problem over the point set $C$ (the $k$-means++ algorithm with $(k+m)$ centers instead of $k$-centers is one such algorithm with $\alpha = O(\log{k})$), and we keep the rest of the algorithm same. 
Hence, we obtain the center-set $F$ of size $(k+m)$, such that $\Phi(F,C) \leq \alpha \cdot \overline{OPT}(C,C)$. Here $\overline{OPT}(C,C)$ denotes the optimal cost for the unconstrained $(k+m)$-service instance $(C,C,k+m,d,\ell)$, i.e., with $(k+m)$ centers. 
Let $\C \coloneqq \{C_{1}, C_{2}, \dotsc, C_{k}\}$ be an optimal clustering for the outlier $k$-service instance defined over $C \setminus Z$, where set $Z$ of size $m$ is the set of outliers. Let $F^{*} \coloneqq \{f_{1}^{*},f_{2}^{*},\dotsc,f_{k}^{*}\}$ denote an optimal center-set for $\C$. Let $\Delta(C_{i})$ denote the optimal cost for each cluster i.e. $\Delta(C_{i}) = \Phi(f_{i}^{*},C_{i})$.
On the basis of previous analysis of the list $k$-service algorithm, we obtain the following property for set $M$.

\begin{quote}{\textbf{Property} \boldmath{$\rm III$}:}\label{property:3}
For any cluster $C_{i} \in \{C_{1}, C_{2}, \dotsc, C_{k}\}$,
there is a point $s_{i}$ in $M$ such that with probability at least $1/2$, the following holds:
\[
    \Phi(t(s_{i}),C_{i}) \leq 
\begin{cases}
    \Big( 3^{\ell} + \fdfrac{\veps}{2} \Big) \cdot \Delta(C_{i}) + \fdfrac{ \veps}{2^{\ell+1} \,k} \cdot \overline{OPT}(C,C),& \textrm{if } C_{i} \in W \\[7pt]
    \Big( 3^{\ell}+\frac{\veps}{2}
    \Big) \cdot \Delta(C_{i}) ,& \textrm{if } C_{i} \in H
\end{cases}
\]
%OPT$(C,C)$, 
where $t(s_{i})$ denotes a facility location that is at least as close to $s_{i}$ as $f_{i}^{*}$, i.e., $d(s_{i},t(s_{i})) \leq d(s_{i},f_{i}^{*})$.
\end{quote}

\noindent The following lemma bounds the optimal $(k+m)$-service cost in terms of the optimal outlier $k$-service cost.
\vspace*{3mm}

\begin{lemma}
$\overline{OPT}(C,C) \leq 2^{\ell} \cdot \sum_{i = 1}^{k} \Delta(C_{i})$
\end{lemma}
\begin{proof}
Let $F_{c} \subseteq C \setminus Z$ denote an optimal center set for $\C$ if centers are restricted to $C\setminus Z$. 
Let us see why $\Psi(F_{c},\C) \leq \sum_{i = 1}^{k} 2^{\ell} \cdot \Delta(C_{i})$. 
Consider each point in $Z$ as a cluster of its own. 
Let us denote these clusters by $\{C_{k+1},C_{k+2}, \dotsc,C_{k+m}\}$. 
Let us define a new clustering $\C' \coloneqq \{C_{1},C_{2},\dotsc,C_{k+m}\}$ having $(k+m)$ clusters. Let $F' \coloneqq F_{c} \cup Z$ be a center-set for $\C'$. Thus we get $\Psi(F',\C') \leq \sum_{i = 1}^{k} 2^{\ell} \cdot \Delta(C_{i})$. 
Moreover, since $\overline{OPT}(C,C)$ is the optimal cost for the unconstrained $(k+m)$-service instance $(C,C,k+m,d,\ell)$ (i.e., with $k+m$ centers), we have $\overline{OPT}(C,C) \leq \Psi(F',\C')$. Hence we get 
$\overline{OPT}(C,C) \leq 2^{\ell} \cdot \sum_{i = 1}^{k} \Delta(C_{i})$.
\end{proof}
\noindent Using the above claim and Property-III, we can say that $\{t(s_{1}),t(s_{2}),\dotsc,t(s_{k})\}$ is $(3^{\ell} + \veps)$-approximation for the outlier $k$-service problem. Hence, we claim that, with probability at least 1/2, there is a center-set in the list that gives $(3^{\ell} + \veps)$-approximation for the outlier $k$-service problem. 
%Also note that, the list size is more than the previous list size, i.e., $((k+m)/\veps)^{O(k \, \ell^{\, 2})}$.
However, the list size is larger in this case. In particular, the list size is $((k+m)/\veps)^{O(k \, \ell^{\, 2})}$.

The outlier $k$-service problem also admits a partition algorithm with running time of $O(nk)$ (see Section~\ref{section: partition_algm}). 
Therefore, we can apply Theorem~\ref{theorem: list_alpha_approx} for the outlier $k$-service problem. This gives us the following main theorem for the outlier problem.

\begin{theorem}
There is a $(3^{\ell}+\veps)$-approximation algorithm for the outlier $k$-service problem with running time of $O(n \cdot ((k+m)/\veps)^{O(k \, \ell^{\,2})})$. 
For the special case, when $C \subseteq L$, the algorithm gives an approximation guarantee of $(2^{\ell}+\veps)$.
\end{theorem}

\noindent It should be noted that for small value of $k$, our algorithm shows improvement over existing algorithms.
The known polynomial time approximation-guarantee for the metric $k$-median problem is $7$ and for the metric $k$-means problem is $53$~\cite{outlier:kmeans_2018_Ravishankar}. 
Our algorithm gives $(3+\veps)$-approximation for the metric $k$-median problem and $(9+\veps)$-approximation for the metric $k$-means problem, albeit in FPT time.
Moreover, the running time of our algorithm is only linear in $n$.

%\section{Streaming Version}
%\label{section:streaming}
%\input{appendix-3-streaming}

\section{Partition Algorithms}
\label{section: partition_algm}
We would like to make some important observations about the constrained problems mentioned in Table~\ref{table:1}. 
The $r$-gather and $r$-capacity $k$-service problems that we consider in this work is more general than the one considered by Ding and Xu~\cite{Ding_and_Xu_15} in that the  $r_{i}$ values for different clusters may not be the same. 
%It differs from the uniform version of the problem, where every cluster has the same $r_{i}$ value. 
Ding and Xu~\cite{Ding_and_Xu_15} only considered the uniform version of the problem (where all $r_i$'s are the same) and designed a partition algorithm for the same. 
However, their algorithm can easily be generalized for the non-uniform version.
We describe these algorithms in the two subsections of this section.
%Section~\ref{section: partition_algm}.
Let us first discuss the last four problems in Table~\ref{table:1} -- fault-tolerant, semi-supervised, uncertain, and outlier $k$-service. These problems do not strictly follow the definition of the constrained $k$-service problem due to which Theorem~\ref{theorem: list_alpha_approx} does not hold for these problems. 
Let us see why.
The fault-tolerant, semi-supervised, and uncertain $k$-service problems have their objective functions different from the one mentioned in Definition~\ref{definition:constrained_k_service}. This distinction is crucial since the list $\mathcal{L}$ only provides a center set, which gives constant approximation with respect to the cost function $\Psi^{*}(\C)$. Therefore, Theorem~\ref{theorem: list_alpha_approx} can not be applied here.
For the outlier $k$-service problem, a feasible clustering is defined over $C \setminus Z$ instead of $C$. Whereas, $\mathcal{L}$ only provide a good center set for the clustering of $C$ and not for $C \setminus Z$. Therefore, Theorem~\ref{theorem: list_alpha_approx} also does not work here.
This problem of misfit can be evaded by defining these problems differently so that they fit the definition of the constrained $k$-service problem. Let us consider these problems one by one.
\begin{enumerate}
    \item An instance of the fault-tolerant $k$-service problem can be reduced to a special instance of the chromatic $k$-service problem (see section 4.5 of~\cite{Ding_and_Xu_15}). Since we can apply Theorem~\ref{theorem: list_alpha_approx} on the reduced instance, we can also obtain the clustering for the original instance.
    \item For the semi-supervised $k$-service problem, we do not need another equivalent definition. In fact, it can be shown that the list $\mathcal{L}$ provides a good center-set for the objective function $\overline{\Psi}(F,\C) \coloneqq \alpha \cdot \Psi(F,\C) + (1-\alpha) \cdot Dist(\C', \C)$ also (see section 4.4 of~\cite{gjk19}). Therefore, we can apply Theorem~\ref{theorem: list_alpha_approx} on a semi-supervised $k$-service instance.
    \item For the uncertain $k$-service problem (\emph{assigned} case), Cormode and McGregor~\cite{uncertain:kcenter_2008_Cormode} gave an equivalent definition as follows: Given a weighted point set $D \coloneqq \cup_{i = 1}^{n} D_{p}$ such that a point $p_{i} \in D_{p}$  carries a weight of $t_{p}^{i}$, the task is to find a clustering $\mathcal{D} \coloneqq \{D_{1}, D_{2}, \dotsc, D_{k}\}$ of $D$ with minimum $\Psi^{*}(\mathcal{D})$, such that \textbf{all} points in $D_{p}$ are assigned to the \textbf{same} cluster. This definition follows from linearity of expectation. Moreover, the weighted instance can easily be converted to an unweighted instance by replacing a weighted point with an appropriate number of points of unit weight (see Section 5~\cite{uncertain:kcenter_2008_Cormode}).
    Now, it fits the definition of the constrained $k$-service problem. Hence, it can be solved using Theorem~\ref{theorem: list_alpha_approx}.
    
    \item At last, we have the outlier $k$-service problem. We do not define the problem differently; instead we construct a new list $\mathcal{L}'$ of center-sets such that it contains a good center-set for any clustering defined over the set $C \setminus Z$ (for any set $Z$ of size $m$).
    However, this new list has a large size, some function of $k$ and $m$.
    Therefore, the FPT algorithm for the outlier $k$-service problem is parameterized by both $k$ and $m$. On the positive side, we only have $k$ in the exponent term.
    We mentioned the details of the algorithm in Section~\ref{section:outlier_list_k_service}.
    Since the problem has not been discussed in~\cite{Ding_and_Xu_15}, we mention its partition algorithm in a subsection here.
\end{enumerate}
We expect that more problems can fit directly or indirectly in the framework of the constrained $k$-service problem. 
Since it already encapsulates a wide range of problems, it is important to study this framework.

Next, we design the partition algorithms for the \emph{non-uniform} $r$-gather, $r$-capacity, and outlier $k$-service problems. 
For rest of the problems mentioned in Table~\ref{table:1}, efficient partition algorithms (in the Euclidean space) are already known (see Section 4 and 5.3 of ~\cite{Ding_and_Xu_15}). 
These algorithms can be easily extended for the general metric space. 

\subsection{\texorpdfstring{$r$}{}-gather \texorpdfstring{$k$}{}-service}
\emph{Partition Problem}: Given a $k$-center-set $F$, we have to find a clustering $\C = \{C_{1}, \dotsc, C_{k}\}$ with minimum $\Psi(F,\C)$, such that each cluster $C_{i} \in \C$ contains at least $r_{i}$ points.

Let us describe the algorithm. Since we do not know which cluster is assigned to which center, let us consider all possible cases. Without loss in generality, assume that $F = \{f_{1}, f_{2}, \dotsc,f_{k} \}$ and  $C_{i}$ is assigned to $f_{i}$. Now, we know that at least $r_{i}$ clients must be assigned to $f_{i}$. Let us reduce the problem to a flow problem on a complete bipartite graph $G = (V_{l},V_{r},E)$. Vertices in $V_{l}$ corresponds to the points in $F$; vertices in $V_{r}$ corresponds to the points in $C$; and an edge $(u,v) \in E$ carries a weight of $d^{\ell}(u,v)$ with a demand $0$ and capacity $1$. Let us create a new source vertex and join it to every vertex in $V_{l}$ and put a demand of $r_{i}$ on each edge. Similarly, create a new sink vertex and join it with every vertex in $V_{r}$ and put a demand of $1$ and capacity of $1$ on each of the edge. It is easy to see that a max-flow on the graph corresponds to a partitioning of $C$ satisfying the $r$-gather problem constraints.
Now we can apply any standard \emph{min-cost flow} algorithm to obtain the optimal partitioning of $C$. 

The running time of the algorithm is $k^{k} \cdot n^{O(1)}$, where $k^{k}$ term appears because we exhaustively guessed the demands for the facilities in $F$, and $n^{O(1)}$ term appears due to the min-cost flow algorithm. Ding and Xu~\cite{Ding_and_Xu_15} used the same method to obtain the partition algorithm for the uniform $r$-gather $k$-service problem. However, there was no guessing involved because each facility had the same demand.

\begin{theorem}
There is a partition algorithm for the $r$-gather $k$-service problem, with a running time of $k^{O(k)} \cdot n^{O(1)}$.
\end{theorem}

\noindent Now, let us discuss the partition algorithm in streaming setting. Goyal \emph{et al.}~\cite{gjk19} designed a streaming partition algorithm for the uniform $r$-gather/$r$-capacity $k$-service problem. 
The same algorithm can be generalized to the non-uniform setting in general metric spaces. 
Let us briefly see what their algorithm does. 
The algorithm first creates a small (in terms of space) representation ($G' = (V_{l}',V_{r}',E')$) of the original graph $G = (V_{l},V_{r},E)$ that satisfies the following definition:

\begin{definition}[Definition 4 of~\cite{gjk19}]
$G' = (V_{l}',V_{r}',E')$ is an edge-weighted bipartite graph such that a number $n_v$ is associated with each vertex $v \in V_{r}'.$ $G'$ {\em represents} the pair $(F, C)$, where $F$ is a set of $k$ centers and $C$ is a set of $n$ points, such that it satisfies the following conditions:
\begin{itemize}
    \item The set $V_{l}' = F$. Each point $p \in C$ is mapped to a unique vertex $v$ in $V_{r}'$ -- call this vertex $\phi(p)$. Further $n_v$ is equal to $|\phi^{-1}(v)|$.
    \item For each point $p \in C$ and center $f \in F$, the weight of the edge $(\phi(p), f)$ in $G'$ is within $(1 \pm \veps)$ of $d^{\ell}(p,f)$. 
\end{itemize}
\end{definition}

\noindent Moreover, the graph $G'$ can be constructed using a two-pass streaming algorithm, as described by the following Theorem.
\begin{theorem}[Theorem 7 of~\cite{gjk19}]
Given a pair $(F,C)$ of $k$ centers and $n$ points respectively, there is a single pass streaming algorithm which builds a bipartite graph $G' = (V_{l}',V_{r}',E')$.
The space used by this algorithm (which includes the size of $G'$) is \\ $O \left(k^2 \cdot 2^k \cdot (k + \log n + \log \Delta) \cdot \left( k \cdot 6^k \cdot  \log{n} + k^k \cdot \log^k (\frac{1}{\veps})  \right) \right)$. Further, the dependence on $\log \Delta$ can be removed by adding one more pass to the algorithm, where $\Delta$ denotes the aspect ratio $\frac{\max_{x \in C, f \in F}d(x,f)}{\min_{x \in C, f \in F}d(x,f)}$ 
\end{theorem}

\noindent It is easy to see that, a flow on the graph $G'$ corresponds to a partitioning of the client set. Moreover, the cost of a flow is at most $(1+\veps)$ times the corresponding partitioning cost. Therefore, we first solve the min-cost flow problem on the graph $G'$. This gives us an optimal flow on graph $G'$. Using it, we can obtain a partitioning of $C$, which takes one additional pass. To know the complete details see Section 4.2 of~\cite{gjk19}. We state the final result as follows.

\begin{theorem}
Consider the partition problem for a $r$-gather $k$-service instance $(C,L,k,d,\ell)$. For a center-set $F \subseteq L$, there is a 3-pass streaming algorithm that outputs a partitioning $\C = \{ C_{1}, C_{2}, \dotsc, C_{k}\}$ of the client set such that the cost of the partitioning is at most $(1+\veps)$ times the optimal partitioning cost. Moreover, the algorithm has space complexity of $f(k,\veps) \cdot \log n$ and the running time of $f(k,\veps) \cdot n^{O(1)}$, where $f(k,\veps) = k^{O(k)} \cdot \log^{k}(\frac{1}{\veps})$.
\end{theorem}

\subsection{\texorpdfstring{$r$}{}-capacity \texorpdfstring{$k$}{}-service}
\emph{Partition Problem}: Given a $k$-center-set $F$, find a clustering $\C = \{C_{1}, \dotsc, C_{k}\}$ with minimum $\Psi(F,\C)$, such that each cluster $C_{i}$ contains at most $r_{i}$ points, for $1 \leq i \leq k$.

The algorithm is similar to the partition algorithm of the $r$-gather $k$-service problem. We can exhaustively consider the capacities of the facilities and then reduce the problem to a flow problem on a complete bipartite graph. The only difference is that the edges connecting the source vertex to $V_{l} = F$, have a demand of $0$ and the capacity of $r_{i}$. Similarly we can design a streaming partition algorithm for the problem. We state the final results as follows.

\begin{theorem}
There is a partition algorithm for the $r$-capacity $k$-service problem, with a running time of $k^{O(k)} \cdot n^{O(1)}$.
\end{theorem}

\begin{theorem}
Consider the partition problem for a $r$-capacity $k$-service instance $(C,L,k,d,\ell)$. For a center-set $F \subseteq L$, there is a 3-pass streaming algorithm that outputs a partitioning $\C = \{ C_{1}, C_{2}, \dotsc, C_{k}\}$ of the client set such that the cost of the partitioning is at most $(1+\veps)$ times the optimal partitioning cost.  Moreover, the algorithm has the space complexity of $f(k,\veps) \cdot \log n$ and running time of $f(k,\veps) \cdot n^{O(1)}$, where $f(k,\veps) = k^{O(k)} \cdot \log^{k}(\frac{1}{\veps})$.
\end{theorem}

\subsection{Outlier \texorpdfstring{$k$}{}-service}
\emph{Partition Problem}: Given a center set $F$, find a set $Z \subseteq C$ of size $m$ and a clustering $\C = \{C_{1}, \dotsc, C_{k}\}$ of $C \setminus Z$ such that $\Psi(F,\C)$ is minimized.
\vspace{2mm}

\noindent The algorithm is simple. First, we define the distance of a point $x$ from the set $F$ as $\min_{f \in F} \left\{d(f,x)\right\}$. Based on these distances, we define the outlier set as a set of $m$ points which are farthest from $F$. The clustering can be obtained by simply running the Voronoi partitioning algorithm on $C \setminus Z$, with $F$ as a center-set. Thus we obtain the following result.

\begin{theorem}
There is a partition algorithm for the outlier $k$-service problem, with the running time of $O(n)$.
\end{theorem}

\noindent The same algorithm can be converted to a streaming algorithm by keeping an account of the farthest point from each center. In the first pass, we can identify the set $Z$. In the second pass, we can output the required clustering, simply by assigning the points to the closest center. We state the final result as follows.

\begin{theorem}
Consider the partition problem for an outlier $k$-service instance $(C,L,k,d,\ell)$. For a center-set $F \subseteq L$, there is a two-pass streaming algorithm that finds a set $Z$ of size $m$ and outputs an optimal partitioning $\C = \{ C_{1}, C_{2}, \dotsc, C_{k}\}$ for $C \setminus Z$. The algorithm has the space complexity of $O(k)$ and the running time of $O(n k)$.
\end{theorem}

\end{document}